%% file: main.tex
\documentclass[letterpaper,journal]{IEEEtran}
\input{common_setting}

\begin{document}
% ===================== Title =====================
\input{component/com-title}
% ===================== Abstrack =====================
\input{component/com-abstract}
% ===================== Keywords =====================
\input{component/com-keywords}
% ===================== Introduction =====================
\input{section/sec-intro}
% ===================== Preliminaries =====================
\input{section/sec-preliminaries}
\input{section/sec-apply2MT}

% ===================== Theoretical Comparison between Using PSS and RS for Selecting Source Test Cases and MGs =====================
\input{section/sec-pssVSrs}
% ===================== Theoretical Comparison on Using PSS between Source Test Case Selection and MG Selection =====================
\input{section/sec-phase1VSphase2}
% ===================== Experiment Setup =====================
\input{section/sec-experimentSetup}
\input{section/sec-results}

\input{section/sec-threats2validity}

% ===================== Related Work =====================
\input{section/sec-relatedWork}
\input{section/sec-conclusion}
% ===================== Acknowledgment =====================
\input{section/sec-acknowledgment}
% ===================== reference =====================
\input{component/com-ref}
% ===================== biography =====================
% \input{component/com-bio}
\end{document}

%% file: common_setting.tex
\usepackage{ifthen}
\usepackage{xcolor}
\usepackage[most]{tcolorbox}
\newboolean{showcomments}
\setboolean{showcomments}{false} % comment this line to deactivate comments
\newboolean{shownewcontent}
\setboolean{shownewcontent}{false}
\ifthenelse{\boolean{showcomments}}{
	\newcommand{\nbc}[3]{
		{\colorbox{#3}{\bfseries\sffamily\scriptsize\textcolor{white}{#1}}}
		{\textcolor{#3}{\sf\small$\langle$\textit{#2}$\rangle$}}}
	
}{
	\newcommand{\nbc}[3]{}

}
\ifthenelse{\boolean{shownewcontent}}{
	\newcommand{\ncc}[1]{{\color{blue}#1}}
  	\newcommand{\ncolor}{\color{blue}}
}{
	\newcommand{\ncc}[1]{#1}
  	\newcommand{\ncolor}{}
}

\newcommand\newcontent[1]{\ncc{#1}}

\newcommand{\RQContent}[1]{%
  \ifcase#1
    null% case 0
  \or Does PSALM empirically outperform RS in terms of fault detection effectiveness?% case 1
  \or How does the effectiveness differ between applying PSALM to source test case selection and to MG selection?% case 2
  \or Besides RS, how does PSALM compare with other representative selection strategies (Adaptive Random Testing (ART) and MT-ART) used in MT?% case 4
  \else null%
  \fi
}

\newcommand{\RQIndex}[1]{RQ\the\numexpr 0 + #1 \relax}
\newcommand{\AnswertoRQ}[2]{
  \begin{tcolorbox}[colback=black!1!white,colframe=black!40!white,
    left=0.88mm, right=0.88mm, top=0.88mm, bottom=0.88mm, boxsep=0.66mm, arc=2.5mm]
    \textbf{Answer to #1:} #2
  \end{tcolorbox}
}
\usepackage{graphicx} 
\usepackage{tabularx}
\usepackage{caption}
\usepackage{subcaption}
\usepackage{array}
\usepackage{multirow}
\usepackage{booktabs}
\usepackage{microtype}

\usepackage{bm}
\usepackage[utf8]{inputenc}
\usepackage[linesnumbered,ruled,vlined]{algorithm2e}
\usepackage{tablefootnote}
\usepackage{amsmath}
\usepackage{amsthm}
\usepackage{amssymb}

\theoremstyle{definition}
\newtheorem{definition}{Definition}

\theoremstyle{plain}

\newtheorem{proposition}{Proposition}
\newtheorem{lemma}{Lemma}

\theoremstyle{remark}

\usepackage{stfloats}        % 放在导言区独立一行

\usepackage[hidelinks]{hyperref}
\usepackage[switch]{lineno}
\usepackage{enumitem}
\setlist[itemize]{topsep=0pt, partopsep=0pt, parsep=0pt, itemsep=0pt}

%% file: component/com-title.tex
\title{PSALM: applying Proportional SAmpLing strategy in Metamorphic testing}

\author{
Zenghui Zhou,
Pak-Lok Poon,~\IEEEmembership{Member,~IEEE,}
% Tsong Yueh Chen,~\IEEEmembership{Member,~IEEE,}
Zheng Zheng*\thanks{* Corresponding author.},~\IEEEmembership{Senior Member,~IEEE,}
Xiao-Yi Zhang% <-this % stops a space
        
\thanks{Z.~Zhou and Z.~Zheng are with the School of Automation Science and Electrical Engineering, Beihang University, Beijing 100191, China.~%
E-mail: \{zhouzenghui,\,zhengz\}@buaa.edu.cn.}%
\thanks{P.-L.~Poon is with the School of Engineering and Technology, Central Queensland University, Melbourne VIC 3000, Australia.~%
E-mail: p.poon@cqu.edu.au.}
\thanks{X.-Y. Zhang is with the School of Computer and Communication Engineering, University of Science and Technology Beijing, Beijing 100191, China.~%
E-mail: xiaoyi@ustb.edu.cn}

}%

\maketitle

%% file: component/com-abstract.tex
\begin{abstract}
Metamorphic testing (MT) alleviates the oracle problem by checking metamorphic relations (MRs) across multiple test executions.
The fault detection effectiveness of MT is influenced not only by the choice and quality of MRs, but also by how source test cases and metamorphic groups (MGs) are selected.
While substantial research has focused on designing, generating, and validating MRs, systematic methods for source test case selection and MG selection remain largely unexplored.
Although the Proportional Sampling Strategy (PSS) provides strong theoretical guarantees in traditional testing, its assumptions cannot be directly applied in MT due to differences in selection domains, test units, and failure distributions.
This paper proposes PSALM, an adaptation of PSS to MT for both source test case selection and MG selection.
We formally prove that PSALM is never inferior to random selection regardless of how the source test case and MG domains are partitioned.
We further identify the conditions under which applying PSALM to source test case selection and MG selection yields identical effectiveness.
A comprehensive empirical study on eight subject programs and 184 mutants shows that the results are consistent with our theoretical analysis and that PSALM generally performs more effectively than existing selection strategies such as ART and MT-ART.
These results demonstrate that PSALM provides a theoretically grounded and practically effective selection strategy for MT.
\end{abstract}

%% file: component/com-keywords.tex
\begin{IEEEkeywords}
  Metamorphic testing, metamorphic relation, software testing, partition testing, proportional sampling strategy.
\end{IEEEkeywords}

%% file: section/sec-intro.tex
\section{Introduction}
\label{sec:intro}
\IEEEPARstart{M}{etamorphic} Testing (MT) is an effective technique to alleviate the \emph{oracle problem} in software testing~\cite{ChenCheung1998, ChanChen1998,Chen2015}, where verifying the correctness of individual system outputs is infeasible or impractical due to: 
(a)~the absence of a mechanism to compute the ``expected'' (i.e., correct) system output of a given input~\cite{BarrHarman2015, SeguraFraser2016}, or
(b)~the presence of resource constraints, thereby prohibiting the human tester from computing the ``expected'' output using the mechanism in~(a).
Unlike ``traditional''  testing techniques (i.e., non-MT) that require knowledge about the expected output of \emph{each individual} input, MT attempts to verify the relationship among \emph{multiple} inputs and their corresponding outputs for fault detection~\cite{ZhouXiang2016, ZhangLiu2024}. The effectiveness of MT in fault detection (in the rest of this paper, ``effectiveness of MT" always refers to its fault detection effectiveness.)
has been demonstrated in various application domains, including embedded systems~\cite{KuoChen2011}, 
compilers~\cite{XiaoLiu2022}, 
machine learning classifiers~\cite{XieHo2011}, 
web services~\cite{ChaleshtariPastore2023,SunWang2011}, 
autonomous car systems~\cite{DengZheng2023, EniserGros2022}, 
large language models~\cite{LiLi2024, HyunGuo2024},
and quantum computing platforms~\cite{PaltenghiPradel2023}.

% In general, MT involves the following four major steps:
% %
% \begin{itemize}
% \item [(1)] \newcontent{\textbf{MR Identification}:} Define metamorphic relations (MRs) based on the properties of the software under test (SUT).
% \item [(2)] \newcontent{\textbf{Source Test Case Selection}:} Select or generate a set of source test cases ($T_s$) such that the size of $T_s$ depends on available testing resources. If more testing resources are available, the size of $T_s$ could be larger.\,%
% \item [(3)] \newcontent{\textbf{MG Construction and Selection}: For each MR, generate follow-up test cases from the selected source test cases. One or more source test cases together with their corresponding follow-up test cases constitute a \emph{metamorphic group} (\emph{MG}).\footnote{\, Thus, an MR is often associated with many MGs.}
%  If the available testing resources do not allow for testing all MGs, then select a representative subset of MGs for testing.}
% \item [(4)] \newcontent{\textbf{MR Verification}:} For each MG selected in Step~3, check whether the execution results with $\mathit{st}$ and $\mathit{ft}_i$ as the inputs violate the relevant MR\@. If yes, then a fault is detected.
% \end{itemize}

In general, MT involves the following five major steps~\cite{ChenCheung1998, Chen2015, SeguraFraser2016}:
\begin{itemize}
\item[(1)] \textbf{Metamorphic relation (MR) identification}: Define MRs based on the properties of the software under test (SUT).
\newcontent{\item[(2)] \textbf{Source test case selection}: Select or generate a set of source test cases (for the rest of this paper, the terms ``test case'' and ``(test) input'' are used interchangeably) where the size of the set depends on available testing resources. If more testing resources are available, the size could be larger.}
%%% \footnote{\,In this paper, the terms ``test case'' and ``(test) input'' are used interchangeably.}
\newcontent{\item[(3)] \textbf{Metamorphic group (MG) construction}: For each MR, generate follow-up test cases from the selected source test cases. One or more source test cases together with their corresponding follow-up test cases constitute an \emph{MG}.
\item[(4)] \textbf{MG selection}: If the available testing resources do not allow all MGs to be tested, then select a representative subset of MGs for testing.}
\item[(5)] \textbf{MR verification}: For each MG selected in Step~4, check whether the execution results with source and follow-up test cases violate the corresponding MR\@.  If a violation is observed, a fault is detected.
\end{itemize}

Obviously, how these steps are performed directly impacts the coverage and, hence, the effectiveness of MT.\@ 
For example, if some MRs (Step~1), ``representative" source test cases (Step~2), or ``representative" MGs (Step~4) are not selected or generated, the comprehensiveness of MT will be adversely affected. 
A close examination of the literature, however, reveals that most existing studies have concentrated on two aspects: 
(a)~the systematic generation of MRs (Step~1)
\cite{SeguraFraser2016, SunFu2021, QiuZheng2022, AyerdiTerragni2024, NolascoMolina2024, ShinPastore2024}, and 
(b)~the applicability and empirical effectiveness of MT in different domains (Step~5)~%
\cite{SeguraFraser2016, ZhouXiang2016, XiaoLiu2022, XieZhang2020, MandrioliShin2025}. 
\newcontent{By contrast, systematic methods for Step~2 (source test case selection) and Step~4 (MG selection) have received less attention, 
even though these two steps are critical to ensuring that MT achieves high fault-detection effectiveness under limited testing resources~\cite{BarusChen2016, SahaKanewala2018, HuiWang2021}.}

\newcontent{There are three key differences between MT and traditional testing:
\begin{itemize}
  \item \textbf{Selection domain.} The selection domain in MT is no longer the original input domain as in the case of traditional testing. MT involves two selection domains: the MR-constrained source test case domain (denoted by $D^{st}$; see Step~2 above) and the derived MG domain (denoted by $D^\mathrm{MG}$; see Step~4 above). In MT, this difference undermines the applicability of those selection methods that directly work on the original input domain~\cite{ChenCheung1998, Chen2015}.

  \item \textbf{Test unit.} In MT, fault detection is performed at the MG (which is a combination of relevant source and follow-up test cases) level. On the other hand, detecting faults is based on individual test cases in traditional testing. In other words, in MT, the \textit{test unit} (defined as the minimal entity used to determine the possible existence of program faults)
  is an MG rather than a single test case as in traditional testing. Accordingly, the focus of MT is failure-causing MGs, whereas the focus of traditional testing is failure-causing test cases.

  \item \textbf{Distribution of failure-causing test units.} The distribution of failure-causing MGs in MT is different from the distribution of failure-causing test cases in traditional testing.

\end{itemize}}

\noindent
\newcontent{Due to these fundamental differences, the theoretical assumptions underlying traditional testing techniques (e.g., the proportional sampling strategy (PSS)~\cite{ChanChen1996, ChenTse2001}) may not hold in the MT setting. As a result, their theoretical conclusions cannot be directly relied upon in MT.}

\newcontent{Motivated by the above observations, we propose PSALM 
(``Proportional SAmpLing strategy in Metamorphic testing''), 
a method that extends PSS to the MT context. 
In  essence, PSALM systematically guides both source test case selection and MG selection by adapting a proportional allocation principle across MR-induced partitions of the source test case domain and the MG domain, aiming to improve the effectiveness of MT\@.}

% Specifically, PSALM addresses the intrinsic differences of MT as follows. 
% First, since the effective selection domain is no longer the original input domain but the MR-constrained source or MG domains, PSALM applies partitioning and proportional allocation over these MR-induced domains. 
% Second, since the test unit in MT is an MG rather than a single input, PSALM performs proportional sampling at the MG level, either by selecting source test cases to generate MGs or by sampling directly from the MG domain.
% Third, since MR constraints reshape the distribution of failure-causing test units, PSALM inherits the robustness of partition testing: by sampling proportionally from all subdomains, its effectiveness is guaranteed to be at least as good as random testing regardless of the underlying failure distribution~\cite{LeungChen1999, LeungChen2000, BaiLee2008}. 
% }

\newcontent{
The main contributions of our work are:
\begin{itemize}
  \item We propose PSALM\,---\, the first method that adapts PSS for source test case selection and MG selection in MT\@. 
  \item We provide formal definitions and proofs that establish the contribution of PSALM to the effectiveness of MT, showing under what conditions PSALM is not inferior to random selection of source test cases and MGs, in terms of fault detection.
  \item We conduct a comprehensive empirical evaluation involving various baseline programs and real-world projects for benchmarking. The empirical results have demonstrated that PSALM outperforms the subject baseline programs in fault detection effectiveness, especially under tight resource constraints.
  % \footnote{\pl{The phrase ``empirical evaluation to confirm theoretical analysis'' is wrong. If the theoretical analysis is correct, we don't need the empirical evaluation. So, I revised this paragraph. (PL)}}
\end{itemize}
}

% In this paper we have developed a method, abbreviated as PSALM (which stands for ``applying {\bf P}roportional {\bf SA}mp{\bf L}ing strategy in {\bf M}etamorphic testing''), involving 
% applying proportional sampling strategy (PSS) to improve the effectiveness in source test case selection (Step~2) and MG selection (Step~4) under tight resource constraints. Note that:
% (i)~the \emph{effectiveness of an MG} is assessed based on its ability to detect faults via a violation of the MR associated to this MG, and
% (ii)~the \emph{effectiveness of a source test case} is assessed based on its ability to generate corresponding follow-up test cases to form \textit{effective} MGs.

\vspace{0.5ex} % adhoc
The rest of the paper is structured as follows.
Section~\ref{motivation} first discusses some issues related to selecting source test cases and MGs in MT, serving as the motivation for developing PSALM\@.
Section~\ref{sec:preliminaries} introduces the preliminaries for understanding PSALM\@.
Section~\ref{sec:PSALM} discusses the major steps of applying PSALM for source test case selection and MG selection.
Section~\ref{sec:pssVSrs} then provides our theoretical analysis to show that PSALM is superior to random selection of source test cases and MGs, in terms of fault detection effectiveness.
This is followed by Section~\ref{sec:phase1VSphase2}, which theoretically analyzes the relative performance of applying PSALM in source test case selection and MG selection.
Section~\ref{sec:experimentSetup} discusses our research questions and experimental setup, aiming to provide further empirical support to the effectiveness of PSALM\@.
Section~\ref{sec:results} discusses and analyzes our experimental results.
Section~\ref{sec:threatsValidity} outlines the threats to the validity of our experiments.
Section~\ref{sec:relatedWork} discusses some related work.
Finally, Section~\ref{sec:conclusion} concludes this paper and discusses some potential future work.

%The rest of the paper is structured as follows.\zzh{add later}
% Section~\ref{motivation} gives the motivation of our study.
% Section~\ref{preliminaries} outlines the main concepts of MT and PSS to facilitate our subsequent discussion.
% Section~\ref{rationale} gives the key rationales underpinning PSALM, and the three main research questions to be addressed in this paper.
% Section~\ref{apply PSS to MT} addresses our first research question, by discussing how to use PSS for selecting source test cases and MGs in the context of MT\@.
% This is followed by Section~\ref{pss_vs_rt}, which provides a theoretical comparison between using PSS and a random approach for selecting source test cases and MGs.
% Section~\ref{phase1_vs_phase2} addresses the second research question, involving a theoretical comparison on using PSS between source test case selection and MG selection.
% Then, Section~\ref{experiment setup} discusses our experimental setup, with a view to further investigating our second and third research questions.
% Section~\ref{results_analysis} discusses and analyzes our experiment results.
% Section~\ref{threats_validity} outlines the threats to validity of our experiment.
% Section~\ref{related work} outlines the related work of our study.
% Finally, Section~\ref{conclusion} concludes our paper and outlines the related future work.

\section{Motivation}
\label{motivation}

\newcontent{
We revisit Steps~2 and~4 of MT (see Section~\ref{sec:intro}), which involve selecting source test cases and MGs, respectively. 
%In practice, these two steps correspond to distinct types of selection problems.
Note that these two steps correspond to distinct types of selection problem, because they involve different selection domains ($D^{st}$ in Step~2 and $D^\mathrm{MG}$ in Step~4) and different selection units (source test cases in Step~2 and MGs in Step~4).

We use a simple example to illustrate the two issues associated with Steps~2 and 4\@.
Consider the $\mathit{sin}(\,)$ program which implements the trigonometric sine function. The input to $\mathit{sin}(\,)$ is an angle $x$ in any degree. Testing $\mathit{sin}(\,)$ involves the oracle problem because, given any $x$, the precise value of its corresponding output (i.e., $\mathit{sin}(x)$) is not known unless $x$ is a special degree (e.g., $x = 90^\circ$ and $\mathit{sin}(x) = 1$). 
There are many properties associated with the sine function; each such property can be defined as an MR\@. An example is MR$_a$ below:

\begin{quote}
Given any input $x$, if we generate another input $y$
such that $y = x + k(360^{\circ})$ (where $k$ is any integer), then $\mathit{sin}(x) = \mathit{sin}(y)$. Here, $x$ is a source test case, $\mathit{sin}(x)$ is a source output, $y$ is a follow-up test case, and $\mathit{sin}(y)$ is a follow-up output.
\end{quote}

Under tight resource constraints, when testing $\mathit{sin}(x)$ with MT, Steps~2 and~4 incur the following issues:
}

\vspace{0.5ex} % adhoc
\begin{itemize}
\item 
\newcontent{\textbf{Explosion of the number of source test cases}: For MR$_a$ mentioned above, the number of source test cases is unbounded, because $x$ can be any angle. Thus, the size of $D^{st}$ is infinite.}
% For instance, consider the function $\mathit{sqrt}(\cdot)$, which computes the square root of a non-negative real number. 
% With respect to an MR stating that for any $x \ge 0$, $\mathit{sqrt}(x^2)$ should equal $\mathit{sqrt}(x)\times\mathit{sqrt}(x)$, the valid source test case domain consists of all non-negative real numbers. 
% Since this domain is continuous and infinite, exhaustive testing is infeasible, and only a limited number of source test cases can be selected under resource constraints.

\item
\newcontent{\textbf{Explosion of the number of MGs}: Given any source test case ($x$), with respect to MR$_a$, the number of corresponding follow-up test cases ($y$) is infinite, because $k$ can be any integer. For example, given $x = 10^\circ$, $y$ can be $370^\circ$ (if $k = 1$),
$730^\circ$ (if $k = 2$),
$1\,090^\circ$ (if $k = 3$), \ldots.
Note that not only the number of follow-up test cases of a source test case (e.g., $x = 10^\circ$) is infinite, but the number of MGs associated with MR$_a$ (and, hence, the size of $D^\mathrm{MG}$) is also infinite, because an MG is a combination of some source test cases and all their corresponding follow-up test cases.
}
\end{itemize}
\vspace{0.5ex} % adhoc

\newcontent{
In view of the above two issues, under resource constraints, with respect to MR$_a$, we have to select a subset of $D^{st}$ and a subset of $D^\mathrm{MG}$ for testing, and this selection is typically performed \textit{randomly} in current MT practice~\cite{BarusChen2016, SunLiu2022}. 
Adopting random selection is mainly because this approach is straightforward to implement and it requires no execution feedback information. On the downside, however, random selection provides no guarantee that the selected test cases or MGs are ``representatives'' of all other elements in $D^{st}$ and $D^\mathrm{MG}$, respectively, thereby potentially reducing the effectiveness of MT~\cite{Gutjahr1999}. 

The above issues have inspired us to develop PSALM, by extending and adapting PSS (a selection approach developed for traditional testing) to the MT context. 
In this paper, we have provided theoretical proofs and conducted experiments that demonstrate that PSALM outperforms the random approach in selecting source test cases and MGs in MT, in terms of fault detection.
% \footnote{\pl{I carefully rephrased this paragraph. I found in several places that you mentioned something like: "PSS guarantees that it is never inferior to random selection \ldots PSALM is never inferior to random selection''. This writing may lead reviewers to think that the proof results of PSALM is similar to that of PSS and, hence, downplay the contribution of our work (PL).}}
}

%% file: section/sec-preliminaries.tex
\section{Preliminaries}
\label{sec:preliminaries}

To facilitate the discussion of PSALM, we first outline the main concepts of MT and PSS\@.

\subsection{Metamorphic Testing}
\label{sec:mt}

MT is a property-based technique developed to address the oracle problem in software testing. 
It checks the relationships between the inputs and outputs of the SUT across multiple executions against the relevant MR, rather than validating the correctness of individual outputs; thus, the oracle problem is alleviated. 
Although MT has been widely used to address the oracle problem in testing, it is also applicable to testing scenarios where the oracle problem does not exist~\cite{ChenCheung1998, SeguraFraser2016, ChenKuo2018}. 
Until now, MT has been effectively applied to a wide range of application domains for fault detection. 

A key concept in MT is MRs, which are the ``expected'' relationships between inputs and outputs across multiple executions of the SUT\@. 
\newcontent{
Given an MR, a set of source test cases ($st_i$s, $i \geqslant 1$) is firstly constructed. Then, for each $st_i$, a corresponding set of follow-up test cases ($ft_{i,j}$s, $i, j \geqslant 1$) is generated. 
Each MG is a set of some $st_i$s associated with their relevant MR, together with all their $ft_{i,j}$s. Thus, each MR is associated with many or even an infinite number of MGs.}

If executing the source and follow-up test cases in an MG produces outputs that violate the MR associated with this MG, this MG is said to trigger a defect (or reveal) a fault in the SUT\@. 
Such an MG is called a \textit{defect-triggering MG}, and its associated MR is called a \textit{defect-triggering MR\@}.

\newcontent{We illustrate the above concept with the $\mathit{sin}(\,)$ program mentioned in Section~\ref{sec:intro}. 
Suppose we test $\mathit{sin}(\,)$ and due to some constraints, the input domain is restricted to $[-360^\circ, 360^\circ]$. 
Consider MR$_b$: For any angles $x$ and $y$ (in degree) such that $x, y \in [-360^\circ, 360^\circ]$ and $(x + y) \in [-360^\circ, 360^\circ]$, we have:
\[
\mathit{sin}(x + y) = \mathit{sin}(x)\mathit{sin}\!\left(90^\circ - y\right)
+ \mathit{sin}(y)\mathit{sin}\!\left(90^\circ - x\right).
\]

Testing MR$_b$ involves the following five executions of $\mathit{sin}(\,)$:

\vspace{0.5ex} % adhoc
\begin{itemize}
\item Execute $\mathit{sin}(\,)$ with the source input $x$ to get the source output $\mathit{sin}(x)$.\,%
\item Execute $\mathit{sin}(\,)$ with the source input $y$ to get the source output $\mathit{sin}(y)$.
\item Execute $\mathit{sin}(\,)$ with the follow-up input $(x + y)$ to get the follow-up output $\mathit{sin}(x + y)$.
\item Execute $\mathit{sin}(\,)$ with the follow-up input $(90^\circ - x)$ to get the follow-up output $\mathit{sin}(90^\circ - x)$.
\item Execute $\mathit{sin}(\,)$ with the follow-up input $(90^\circ - y)$ to get the follow-up output $\mathit{sin}(90^\circ - y)$.
\end{itemize}

\vspace{0.5ex} % adhoc
Here, we can form an MG$_b$ (corresponding to MR$_b$) containing the above five source and follow-up inputs: $x$, $y$, $(x + y)$, $(90^\circ - x)$, and $(90^\circ - y)$.
With the above five source and follow-up outputs, we then check whether or not MR$_b$ is violated.
If yes, then MG$_b$ is a defect-triggering MG, and MR$_b$ is a defect-triggering MR\@.
}

\newcontent{Several points should be noted from the above example involving testing $\mathit{sin}(\,)$ with MT:

\begin{itemize}
\item For the first two executions above, the source inputs $x$ and $y$ can be randomly selected from the input domain. However, when testing some other programs, we may need to generate source test cases using a systematic method such as the choice relation framework~\cite{ChenPoon2003, PoonTang2010} or the classification-tree methodology~\cite{ChenPoon2000}.

\item Given $x$ and $y$, their corresponding follow-up inputs $(x + y)$, $(90^\circ - x)$, and $(90^\circ - y)$ can be systematically derived according to MR$_b$. 

\item Checking whether or not MR$_b$ holds does not require the knowledge of the correct \textit{individual} outputs of $\mathit{sin}(x)$, $\mathit{sin}(y)$, $\mathit{sin}(x+y)$, $\mathit{sin}\!\left(90^\circ - x\right)$, and $\mathit{sin}\!\left(90^\circ - y\right)$.

\item The tuple $\langle x, y, x + y, 90^\circ - x, 90^\circ - y \rangle$ constitutes MG$_b$ with respect to MR$_b$. 
Since $x$ and $y$ can take any values from the input domain $[-360^\circ, 360^\circ]$ provided that the follow-up test cases derived from them, such as $(x+y)$, also remain within this domain, MR$_b$ is associated with an infinite number of MGs.

\end{itemize}
}

\vspace{0.5ex} % adhoc
\newcontent{
To illustrate the diversity of metamorphic relations, consider MR$_c$:
\[
\mathit{sin}(-x) = -\mathit{sin}(x),
\]
which captures the odd symmetry of the sine function. 
With respect to MR$_c$, given a source input $x$, the corresponding follow-up input is $-x$, and the resulting MG is $\langle x, -x \rangle$. 
MR$_c$ differs from MR$_b$ in that the former involves only one source input and only one follow-up input. This illustrates that MRs can vary in complexity and structure. 
%%% In subsequent sections, MR$_a$ and MR$_b$ will be used to demonstrate how PSALM applies PSS to select MGs across multiple MRs.
}

\subsection{Proportional Sampling Strategy (PSS)}
\label{sec:pss}

\emph{Partition testing} is a technique that divides the input domain (denoted by $D$) into input subdomains\,%
\footnote{\,In the rest of the paper, when there is no ambiguity, we will simply refer to ``input subdomains'' as ``subdomains''.}
(also called partitions if these subdomains are disjoint), and test cases are selected from each subdomain for software execution. The partition scheme is based on some predefined criteria typically derived from the specification or the program code of SUT~\cite{GrochtmannGrimm1993,BohmePaul2016}.
The main idea behind partition testing is that the subdomains so generated are expected to be ``homogeneous'', which means that all test cases in a subdomain should cause the SUT either to succeed or to fail. In other words, any test case selected from a subdomain $D_i$ ($i \geqslant 1$) can be considered ``representative'' of all other test cases in $D_i$. 
Therefore, instead of using every test case in $D_i$ to test the SUT, only a subset of test cases is needed to effectively detect faults associated with $D_i$~\cite{ZhengXu2013,BaiLee2008}. Because of this characteristic, partition testing is often used for test case selection.

The comparison between random testing and partition testing has often been discussed until the \emph{Proportional Sampling Strategy} (\emph{PSS}) was proposed.
PSS suggests selecting the number of test cases from each $D_i$ in proportion to the size of $D_i$~\cite{LeungChen1999,LeungChen2000}.\,%
\footnote{\,In partition testing and PSS, the generated subdomains can be overlapping. However, many existing studies on PSS assume that the subdomains are disjoint. Our method PSALM, discussed in this paper also takes this assumption.}
Studies~\cite{ChanChen1996, ChenTse2001} have formally proved that the fault detection effectiveness of PSS is equal to or higher than random sampling (RS).
To perform PSS, the tester requires the information about the size ratios of the subdomains for deciding the number of test cases to be selected from each subdomain. This information is often available, rendering PSS a practical and effective software testing technique. 

\newcontent{In traditional testing, PSS has been proven to guarantee the fault detection effectiveness over RS regardless of the specific partition scheme used. 
It is this feature that contributes to the high impact of PSS in traditional testing.
This work investigates whether such a guarantee still holds in the context of MT, where the selection domains, test units, and failure distributions fundamentally differ from those in traditional testing.
}

%% file: section/sec-apply2MT.tex
\section{PSALM: applying Proportional SAmpLing strategy in Metamorphic testing}
\label{sec:PSALM}
\newcontent{

PSALM has a similar independence feature as PSS\@, i.e., when applying PSALM for selecting source test cases from $D^{st}$ and selecting MGs from $D^{\mathrm{MG}}$, the effectiveness of MT does not depend on how 
$D^{st}$ and $D^{\mathrm{MG}}$ are partitioned. This feature makes applying PSALM highly versatile in MT\@.
It is also the reason why the discussion on partitioning $D^{st}$ and $D^{\mathrm{MG}}$ falls beyond the scope of this paper.
}

\newcontent{The proofs of the propositions that underpin PSALM will be given in Section~\ref{sec:pssVSrs}.}
We first discuss the application of PSALM for selecting source test cases, and then move on to discuss applying PSALM for selecting MGs in this section.
% \footnote{\pl{After a second thought, I think saying like ``applying PSS for selecting source test cases/MGs'' is no good. It implies that we ``directly'' apply PSS (that is, without change) to the selection. It will significantly downplay the contribution of PSALM. This may be the reason why reviewers wonder why we need to redo the proof again. Instead, we should say ``applying PSALM for selecting source test cases/MGs. (PL)}}

\subsection{Applying PSALM for Source Test Case Selection}
\label{sec:applyPSS4stcs}

% When comparing the use of PSS for selecting (source) test cases in MT and in traditional testing, the main difference lies in the source of selection
% and the approach to devising the partition scheme. Once partitioning has been completed, the subsequent selection of test cases between MT and traditional testing is similar.

\newcontent{As explained in Section~\ref{motivation}, given an MR, we often encounter the explosion of the number of source test cases. Under resource constraints, testing all these source test cases (and their corresponding follow-up test cases) is obviously infeasible. This is the first challenge that PSALM attempts to address.

Recall that $D^{st}$ denotes the source test case domain of a given MR.
Applying PSALM for source test case selection with a view to forming $D^\mathrm{MG}$ for MR involves the following major steps:
}

% \begin{figure*}[!htbp]
%   \centering
%   \includegraphics[width=0.9\linewidth]{figure/fig:pssVSrs4stcs.png}
%   \caption{Comparing using PSS for selecting test cases in traditional testing (a) and MT (b)\@.}
%   \label{fig:pssVSrs4stcs}
% \end{figure*}

\vspace{0.5ex} % adhoc
\begin{itemize}
  \item [(1)] \newcontent{Let $\mathcal{C}$ denote the intrinsic constraint of the input domain $D$, 
  and $\mathcal{C}_\mathrm{MR}$ denote the additional constraint induced by MR (i.e., the requirement 
  that all follow-up test cases generated from any $\mathit{st}$ must also fall within a certain valid input range). 
  With reference to both $\mathcal{C}$ and $\mathcal{C}_\mathrm{MR}$, derive the feasible source test case domain $D^{st} \subseteq D$.}
  
  \item [(2)] Partition $D^{st}$ into subdomains $\{D^{st}_1, D^{st}_2, \dots, D^{st}_k\}$ ($k \geqslant 1$), by considering the information obtained from the specification and the parameters of the SUT (as in traditional testing), as well as MR.

  %\newcontent{\item[(2)] Determine the number of source test cases ($\mathit{st}$s) to be selected from each %subdomain. 
  %The number of $\mathit{st}$s selected from a subdomain should be proportional to the size of this subdomain %relative to the size of $D^{st}$.}

  \item[(3)] Select $\mathit{st}$s from each $D^{st}_i$ ($i \geqslant 1$) such that the number of $\mathit{st}$s selected from $D^{st}_i$ should be proportional to the size of $D^{st}_i$ relative to the size of $D^{st}$.
  
\newcontent{\item[(4)] For each selected $\mathit{st}$, with reference to MR, generate its corresponding follow-up test cases ($\mathit{ft}$s).}

% \newcontent{\item[(5)] 
% Check whether all the $\mathit{ft}$s generated in Step~4 satisfy $\cal C$ and  $\mathcal{C}_\mathrm{MR}$. If yes, proceed to the next step. Otherwise, go back to Step~1 to adjust the ``initial'' $D^{st}$, and then repeat Steps~2 to~5 until all the $\mathit{ft}$s satisfy $\cal C$ and $\mathcal{C}_\mathrm{MR}$. 
% \footnote{\tt This step has to be kept. Otherwise, the procedure is wrong. For the subsequent example, it is NOT a must that it needs to show every step. (PL)\zzh{If the “initial’’ $D^{st}$ is already derived by considering both $\mathcal{C}$ and $\mathcal{C}_{\mathrm{MR}}$, why is Step~5 still necessary? Since the constraints ensuring that the follow-up test cases satisfy $\mathcal{C}$ and $\mathcal{C}_{\mathrm{MR}}$ are already included in $\mathcal{C}_{\mathrm{MR}}$, it seems redundant to perform Step~5.}}
% }

\newcontent{\item[(5)] Generate a set of MGs (i.e., $D^\mathrm{MG}$) for MR. Each MG contains multiple $\mathit{st}$s selected in Step~3 and their corresponding $\mathit{ft}$s generated in Step~4.
}
\end{itemize}

\vspace{0.5ex}
\newcontent{
We illustrate how to apply the above steps using the $\mathit{sin}(\,)$ program introduced in 
Section~\ref{sec:mt}. 
For this program, the intrinsic constraint $\mathcal{C}$ requires every input of $\mathit{sin}(\,)$ to lie 
within $[-360^\circ, 360^\circ]$. 
With respect to MR$_b$, the MR-induced constraint $\mathcal{C}_{\mathrm{MR}}$ requires that all follow-up 
test cases (e.g., $(x+y)$, $(90^\circ - x)$, and $(90^\circ - y)$) must also stay within this valid range. 
To satisfy these constraints, the source test cases $x$ and $y$ must satisfy $x, y \in [-180^\circ, 180^\circ]$, 
and therefore the feasible range of $D^{st}$ becomes $[-180^\circ, 180^\circ]$.

Suppose that, in view of the available testing resources, a total of 30 test cases are to be selected. 
In other words, a maximum of 30 program executions can be performed. 
Based on the feasible range, $D^{st}$ is partitioned into three subdomains: 
$D^{st}_1 = [-180^\circ, -90^\circ)$, 
$D^{st}_2 = [-90^\circ, 0^\circ)$, and 
$D^{st}_3 = [0^\circ, 180^\circ]$.
}

\newcontent{
Here, $D^{st}_1$ and $D^{st}_2$ have equal sizes of $90^\circ$, whereas $D^{st}_3$ has a larger size of $180^\circ$. 
Accordingly, the number of source test cases selected from
$D^{st}_1$, $D^{st}_2$, and $D^{st}_3$ should be in the ratio $1:1:2$.

Since only 30 test executions are allowed, and each MG (with respect to MR$_b$) involves 2 source test cases ( $x$ and $y$ ) and 3 follow-up test cases
($(x + y)$, $(90^\circ - x)$, and $(90^\circ - y)$) (see Section~\ref{sec:mt}), 
at most 6 MGs can be executed. 
Accordingly, a total of 12 ($= 2 \times 6$) source test cases need to be selected from the adjusted $D^{st}$.
Thus, in Step~3, considering the relative size of the three subdomains (i.e., $1 : 1 : 2$), the numbers of source test cases selected from $D^{st}_1$, $D^{st}_2$, and $D^{st}_3$ are 3, 3, and 6, respectively. Also, in Step~4, a total of 18 ($= 3 \times 6$) follow-up test cases are generated from their corresponding source test cases.
}

In some occasions, in Step~4, the relative size of $D_i^\mathit{st}$s may lead to the situation where the number of source test cases to be selected from some $D_i^\mathit{st}$s is not an integer. 
We have developed the \newcontent{Basic Maximin Algorithm for MT (BMA-MT)} to address this issue. 
The essence of this algorithm is to achieve a better ``balanced'' distribution of selected test units (e.g., source test cases) among different selection subdomains (e.g., $D_i^\mathit{st}$s) by methodically selecting test units from the selection subdomain whose current sampling rate is the \textit{lowest}. 
However, when there are two or more selection subdomains with the same lowest sampling rate, BMA-MT suggests the \textit{next} test unit to be selected from one of these selection subdomains whose size is the \textit{largest}. 
\newcontent{If multiple such subdomains share the same largest size, then one of them is randomly selected.}

Algorithm~\ref{alg:BMA-MT} shows the detailed steps of \newcontent{BMA-MT}, in which $d$ denotes the size of the selection domain, $k$ denotes the number of selection subdomains, and $n$ denotes the total number of test units selected from the selection domain. 
Furthermore, for any $i$th ($1 \leqslant i \leqslant k$) selection subdomain, $d_i$ denotes its size; $n_i$ denotes the number of source test cases selected from it, and $\sigma_i$ denotes its sampling ratio.

\begin{algorithm}[!htbp]
  \caption{The Basic Maximin Algorithm for MT (BMA-MT)} \label{alg:BMA-MT}
  \SetKwComment{Comment}{/* }{ */}
   \textbf{Initialization:} $n_i \leftarrow 1$, $\sigma_i \leftarrow \frac{1}{d_i}$, $q \leftarrow n-k$\;
    \While{q $>$ 0}{
    \newcontent{
      \If{$\bigl|\operatorname*{argmin}_i (\sigma_i)\bigr| > 1$}{
          $\mathcal{M} \leftarrow \operatorname*{argmin}_i (\sigma_i)$\;
          \If{$\bigl|\operatorname*{argmax}_{m \in \mathcal{M}} (d_m)\bigr| > 1$}{
              $j \xleftarrow{\text{rand}} \operatorname*{argmax}_{m \in \mathcal{M}} (d_m)$\;
          }
          \Else{
              $j \leftarrow \operatorname*{argmax}_{m \in \mathcal{M}} (d_m)$\;
          }
      }
      \Else{
          $j \leftarrow \operatorname*{argmin}_i (\sigma_i)$\;
      }
    }
      $n_j \leftarrow n_j + 1$\;
      $\sigma_j \leftarrow \sigma_j + \frac{1}{d_j}$\;
      $q \leftarrow q - 1$\;
    }
    
\end{algorithm}

\par \newcontent{Consider an example where we need to select 25 test cases, among which 10 are source test cases selected from $D^{st}$. The relative size ratio of the three selection subdomains is $1:1:2$. According to BMA-MT:}

\newcontent{
\begin{itemize}
    \item[(1)] First, one source test case is selected from each selection subdomain.
    \item[(2)] When selecting the fourth source test case, $D^{st}_3$ has the lowest sampling rate. 
    Therefore, the fourth source test case is selected from $D^{st}_3$.
    \item[(3)] When selecting the fifth source test case, all subdomains have the same sampling rate. 
    Since $D^{st}_3$ has the largest size, the fifth source test case is selected from $D^{st}_3$.
    \item[(4)] After selecting the fifth source test case, $D^{st}_1$ and $D^{st}_2$ have the same lowest sampling rate and the same size. 
    Thus, the sixth source test case is arbitrarily selected from either $D^{st}_1$ or $D^{st}_2$.
    \item[(5)] Repeat the above steps until a total of 10 source test cases are selected. 
    At the end of the process, the number of selected source test cases can be: 
    $D^{st}_1$: 3, $D^{st}_2$: 2, and $D^{st}_3$: 5.
\end{itemize}
}

As a side note for BMA-MT, when the relative size of the selection subdomains causes the number of source test cases to be selected from each subdomain to be an integer, applying BMA-MT for selecting source test cases will produce the same result as directly selecting these source test cases from the selection subdomains based on the ratio of their sizes, without applying BMA-MT\@.

% Also, BMA-MT is an adaptation (with improvement) of the ``original'' BMA algorithm~\cite{ChenYu2001}. In the ``original'' BMA algorithm, when multiple selection subdomains ($d_1, d_2, \ldots, d_k$) have the same lowest sampling rate, the next test case is \textit{randomly} selected from $d_i$ ($1 \leqslant i \leqslant k$). In contrast, BMA-MT recommends selecting the next test case from $d_i$ ($1 \leqslant i \leqslant k$) with the largest size. This recommendation aligns with the rationale behind PSS, which states that more test cases should be selected from those selection subdomains with larger sizes.

\subsection{Applying PSALM for MG Selection}
\label{sec:applyPSS4MGs}

\newcontent{When generating MGs from source test cases under multiple MRs, the total number of MGs can be very large.
Executing all MGs is often infeasible; thus, a subset of MGs has to be selected.}
For MG selection, the selection is done on a set of generated MGs (called the \textit{MG domain},
and is denoted by $D^\mathrm{MG}$), with respect to all the defined MRs.

% In traditional testing, the focus of applying PSS is on selecting test cases \emph{individually} from subdomains.
In MT, besides using PSSLM for selecting source test cases (as discussed in Section~\ref{sec:applyPSS4stcs}), the selection focus could also be on MGs. In other words, the selection units are no longer source test cases, but the MGs.

% Fig.~\ref{fig:pssVSrs4mgs} further illustrates the main differences between selecting individual test cases in traditional testing (Fig.~\ref{fig:pssVSrs4mgs}(a)) and selecting MGs in MT (Fig.~\ref{fig:pssVSrs4mgs}(b)).
% Fig.~\ref{fig:pssVSrs4mgs}(b) also shows that, when selecting MGs from $D^\mathrm{MG}$, an additional dimension\,---\,MRs\,---\,has to be considered (something that is irrelevant to traditional testing).

% \begin{figure*}[!htbp]
%   \centering
%   \includegraphics[width=\linewidth]{figure/fig:pssVSrs4MGs.png}
%   \caption{Comparing using PSS for selecting test cases in traditional testing (a) and selecting MGs in MT (b)\@.}
%    \label{fig:pssVSrs4mgs}
% \end{figure*}

Below we list the major steps of using PSALM to select MGs from $D^\mathrm{MG}$ for all the defined MRs (assuming that there are $k$ defined MRs):

\vspace{0.5ex} % adhoc
\begin{itemize}
%   \item [(1)] Randomly select $m$ test cases from $D$ (whose size is often infinite) to form $D'$ ($D' \subsetneq D$). The idea is to ensure that the size of $D'$ ($= m$) is finite.
  
    \item[(1)] Repeat this step until all the defined MRs have been processed:
    \begin{itemize}
    \item[(a)] Select any unprocessed MR (denoted by MR$_i$, where $1 \leqslant i \leqslant k$).
    \item[(b)] For all source test cases relevant to MR$_i$, generate the corresponding follow-up test cases according to MR$_i$.
    \newcontent{\item[(c)] Combine the relevant source and follow-up test cases to form the MGs with respect to MR$_i$. All these MGs collectively constitute an \textit{MG subdomain related to MR$_i$}, denoted by $D^{\mathrm{MG}}_{\mathrm{MR}_i}$ ($1 \leqslant i \leqslant k$).}
    \end{itemize}
  
   \item [(2)] Generate $D^\mathrm{MG}$ for all the $k$ defined MRs such that
   \newcontent{$D^{\mathrm{MG}} = D^{\mathrm{MG}}_{\mathrm{MR}_1} \cup D^{\mathrm{MG}}_{\mathrm{MR}_2} \cup \cdots \cup D^{\mathrm{MG}}_{\mathrm{MR}_k}.$}

   \item [(3)] Partition $D^{\mathrm{MG}}$ by considering the specification and parameters of the SUT, as well as the defined MRs.
        
  \item [(4)] Determine the number of MGs to be selected from each partitioned MG subdomain based on the relative sizes of these subdomains.
  
  \item [(5)] Select the predetermined number of MGs from each MG subdomain according to Step~4\@. All the selected MGs (from various MG subdomains) together form the MG test set.
\end{itemize}

\vspace{0.5ex} % adhoc
\newcontent{
We illustrate how to apply the above steps using the $\mathit{sin}(\,)$ example mentioned in Section~\ref{sec:mt}.
Suppose that, for this program, two MRs have been defined: MR$_b$ and MR$_c$ (as described in Section~\ref{sec:mt}).
Due to resource constraints, we consider 800 source test cases for each MR.

For MR$_b$, each MG is constructed from two source test cases together with three follow-up test cases.
Therefore, these 800 source test cases can be paired to form 400 MGs. 
As for MR$_c$, each source test case is associated with exactly one follow-up test case, and hence each source test case directly corresponds to one MG.
As a result, the 800 source test cases yield 800 MGs.
}

\newcontent{Considering MR$_b$ and MR$_c$ together, we have a total of 1{,}200 MGs, which could be too resource demanding to test all of them. Suppose the available testing resources allow us to select 30 MGs out of the 1{,}200 for testing.
Below we illustrate how to apply PSS to do MG selection:
}

\vspace{0.5ex} % adhoc
\newcontent{
\begin{itemize}
    \item[(1)] We have two MG subdomains: one ($D^{\mathrm{MG}}_{\mathrm{MR}_b}$) is related to MR$_b$ and the other one ($D^{\mathrm{MG}}_{\mathrm{MR}_c}$) is related to MR$_c$.     \vspace{0.2em}
    As stated above, the size of $D^{\mathrm{MG}}_{\mathrm{MR}_b}$ and $D^{\mathrm{MG}}_{\mathrm{MR}_c}$ is 400 and 800, respectively.  

    \item[(2)] Generate the MG domain ($D^{\mathrm{MG}}$) as follows:
    \[
    D^{\mathrm{MG}} = D^{\mathrm{MG}}_{\mathrm{MR}_b} \cup D^{\mathrm{MG}}_{\mathrm{MR}_c},~\mathrm{where} |D^{\mathrm{MG}}| = 1{,}200.
    \]

    \item[(3)] Define a partition scheme for $D^{\mathrm{MG}}$.  
    Suppose that, using the defined partition scheme, $D^{\mathrm{MG}}$ is divided into MG subdomains of equal size (i.e., 300 each): $D^{\mathrm{MG}}_{1}$, $D^{\mathrm{MG}}_{2}$, $D^{\mathrm{MG}}_{3}$, and $D^{\mathrm{MG}}_{4}$.

    \item[(4)] Since only 30 MGs are to be selected, PSS is applied to allocate the number of MGs to be selected from each MG subdomain, by considering the size of this MG subdomain relative to that of other MG subdomains.
    Given that the relative sizes among $D^{\mathrm{MG}}_{1}$, $D^{\mathrm{MG}}_{2}$, $D^{\mathrm{MG}}_{3}$, and $D^{\mathrm{MG}}_{4}$ are 1:1:1:1, the allocation is $30/4 = 7.5$ MGs per MG subdomain.  
    After applying BMA-MT, 8 MGs are selected from two MG subdomains and 7 from the other two.  

    \item[(5)] The selected 30 MGs now form the MG test set for test execution and analysis.
\end{itemize}
}

\vspace{0.5ex} % adhoc
\newcontent{
Although source test case selection and MG selection can both be based on PSS, these selection processes have some fundamental differences. 
First, in terms of selection domain, $D^{st}$ contains individual source test cases extracted from $D$, whereas $D^\textrm{{MG}}$ contains tuples of source and follow-up test cases generated according to the defined MRs. 
Second, in terms of unit of selection, source test case selection focuses on individual source test cases, whereas MG selection focuses on individual MGs. Note that each MG involves source test cases together with their follow-up test cases. In other words, each MG selection does not only involve source test cases. Third, regarding the selection sequence, MG selection can only be performed until source test case selection is finished, because follow-up test cases can only be generated based on the selected source test cases.
}

%% file: section/sec-pssVSrs.tex
\section{Theoretical Comparison between PSALM and RS in Selecting Source Test Cases and MGs}
\label{sec:pssVSrs}

\newcontent{
After illustrating the application of PSALM to source test case selection and MG selection with examples, we now provide theoretical analyses on the effectiveness of PSALM when compared to RS. 
This section establishes the mathematical foundations, showing that PSALM has the same or higher effectiveness than RS, in terms of fault detection using MT\@.
The symbols used in the subsequent analysis are summarized in Table~\ref{tab:symbols}.
}

\subsection{Using PSALM and RS for Source Test Case Selection}
\label{sec:pssVSrs4stcs}
In traditional testing, the notion of test case effectiveness is determined by the ability of a test case to trigger or reveal defects in the SUT.
This notion, however, cannot be directly applied to MT\@. Due to the oracle problem, we are unable to determine the correctness of the output generated by an individual test case. It is this reason why MT involves multiple executions so that we can check the relationship between source/follow-up test cases and source/follow-up outputs against an MR to determine the existence of a fault. 
A prerequisite of the original proof of PSS in traditional testing is the ability to determine whether or not an individual test case is failure-causing.
In the context of MT, however, failure cannot be determined at the level of individual source test cases.
Hence, the original proof of PSS is not directly applicable to MT when selecting source test cases.
Therefore, we need to develop a dedicated proof for applying PSALM to source test case selection in MT in this subsection.

\newcontent{First, we need to give a formal definition related to the ``contribution'' of a source test case to fault detection via the failure-causing probabilities of its associated MGs.}

\begin{table}[!ht]
    % \footnotesize
    \centering
    \renewcommand{\arraystretch}{1.2}  
    \caption{Notation used in the theoretical analysis}
    \label{tab:symbols}
    \resizebox{!}{!}{
        \input{tables/tab-symbols}
    }
\end{table}
   
\begin{definition}[\bf Effectiveness of a source test case in revealing defects]
\label{def:stEffectiveness}
In MT, the effectiveness of a source test case (\textit{st}) is defined as the probability that the MGs associated with \textit{st} reveal defects (denoted by $\mathit{vp}(\mathit{st})$):
\[ vp(st) = \frac{\left| \mathrm{MG}_{st}^v \right|}{\left| \mathrm{MG}_{st} \right|},\]
where $\mathrm{MG}_{st}$ denotes the set of MGs constructed from \textit{st}, $\mathrm{MG}_{st}^v$ denotes the subset of those MGs that violate their corresponding MR, and $\left| * \right|$ denotes the size of a set.
\end{definition}

\newcontent{Given a source test case subdomain $D^{st}_i$, its fault detection effectiveness can be measured in terms of the mean fault detection effectiveness of all the source test cases in $D^{st}_i$, which reflects the probability of revealing a defect by any source test case randomly selected from $D^{st}_i$. This rationale is formally expressed as Definition~\ref{def:subdomainEffectiveness} below:}

\begin{definition}[\bf Effectiveness of a source test case subdomain in revealing defects]
\label{def:subdomainEffectiveness}
The effectiveness of $D^{st}_i$ to trigger defects (denoted by $\mathit{vr}_i$) is defined as:
\[
    vr_i=\frac{\sum_{j=1}^{d_i} vp\left(st_j\right)}{d_i}.
\]
\end{definition}

\newcontent{Similar to Definition~\ref{def:subdomainEffectiveness} above, we can define the fault detection effectiveness of the entire source test case domain ($D^{st}$) in terms of the probability that a source test case randomly chosen from $D^{st}$ reveals a defect. This probability is computed by averaging the fault detection effectiveness of all the source test cases in $D^{st}$. This leads to Definition~\ref{def:selectiondomainEffectiveness} below:}

\begin{definition}[\bf Effectiveness of a source test case domain in revealing defects]
\label{def:selectiondomainEffectiveness}
The effectiveness of a selection domain to trigger defects (denoted by $\mathit{vr}$) is defined as:
\[
    vr=\frac{\sum_{j=1}^{d} vp\left(st_j\right)}{d}.
\]
\end{definition}

\newcontent{In what follows, we adopt the standard probabilistic model used in prior work on the P-measure~\cite{Ntafos2001, AhlgrenBerezin2021, MayerGuderlei2006} for our proofs. 
First, we make the following two assumptions:

\vspace{0.5ex} % adhoc
\begin{itemize}
\item {\bf Assumption~1:} Selecting source test cases from $D^{st}$ and $D^{st}_i$
is assumed to be a selection with replacement. This assumption is made because the number of source test cases selected is often very small when compared with the size of $D^{st}$.
\item {\bf Assumption~2:} All the selection subdomains are non-overlapping. Note that, even if there are some overlapping subdomains, we can further divide these subdomains to make them non-overlapping so that this assumption can be met.
\end{itemize}

\vspace{0.5ex} % adhoc
Under the above two assumptions, we formalize the effectiveness of using PSALM and RS for source test case selection as follows\footnote{In the rest of the paper, unless otherwise stated, ``the effectiveness of PSALM and RS'' refers to the effectiveness of using PSALM and RS for source test case or MG selection in MT for fault detection.}.}

\begin{definition}[\bf Effectiveness of using PSALM for source test case selection]
\label{def:PPss}
The effectiveness of using PSALM for source test case selection (denoted by $P_{p}^{\mathit{st}}$) is defined as:
\[
P_{p}^{st} = 1 - \prod_{i=1}^k (1 - vr_i)^{n_i}.
\]
\end{definition}

\begin{definition}[\bf Effectiveness of using RS for source test case selection]
\label{def:PRT}
The effectiveness of using RS for source test case selection (denoted by $P_{r}^{\mathit{st}}$) is defined as:
\[
P_{r}^{st} = 1 - (1 - vr)^n.
\]
\end{definition}

Based on Definitions~\ref{def:PPss} and~\ref{def:PRT}, we have the following proposition:

\begin{proposition}[\bf Effectiveness comparison between using PSALM and RS for source test case selection]
\label{prop:EffectivenessComparisonSt}
Regardless of the partition scheme used, we have $P_{p}^{st} \geqslant P_{r}^{st}$.
\end{proposition}

\begin{proof}
We apply mathematical induction to prove Proposition~\ref{prop:EffectivenessComparisonSt}. We first consider the simplest case, which involves two selection subdomains (i.e., $k=2$). 
According to Definitions~\ref{def:PPss} and~\ref{def:PRT}, we have:
\[
P_{p}^{st} = 1 - (1 - vr_1)^{n_1} (1 - vr_2)^{n_2}, 
~
P_{r}^{st} = 1 - (1 - vr)^{n}.
\]

Under PSALM, test cases are selected in proportion to the relative size of the subdomains, and thus we have $n_1/d_1 = n_2/d_2$. Suppose $d_1\geqslant d_2$ and $r=n_1/n_2=d_1/d_2$. Then $r \geqslant 1$, and we have:
\begin{align*}
vr_2-vr=\frac{\sum_{j=1}^{d_2} vp\left(st_j\right)}{d_2}-\frac{\sum_{j=1}^{d} vp\left(st_j\right)}{d} \\
vr-vr_1=\frac{\sum_{j=1}^{d} vp\left(st_j\right)}{d}-\frac{\sum_{j=1}^{d_1} vp\left(st_j\right)}{d_1}.
\end{align*}

With the above assumptions, we have:
\[
D_1 \cup D_2 = D, ~ D_1 \cap D_2 = \emptyset,
\]

\noindent
which leads to the following:
\begin{align*}
vr_2-vr&=\frac{\sum_{j=1}^{d_2} vp\left(st_j\right)}{d_2}-\frac{\sum_{j=1}^{d} vp\left(st_j\right)}{d}\\
&=\frac{\sum_{j=1}^{d_2} vp\left(st_j\right)}{d_2}-\frac{\sum_{j=1}^{d_1} vp\left(st_j\right)+\sum_{j=1}^{d_2} vp\left(st_j\right)}{d_1 +d_2}\\
&=\frac{d_1\sum_{j=1}^{d_2} vp\left(st_j\right)-d_2\sum_{j=1}^{d_1} vp\left(st_j\right)}{d_2(d_1+d_2)}\\
&= \frac{d_1}{d_2}
   \left(
     \frac{\sum_{j=1}^{d_1} vp(st_j)
           + \sum_{j=1}^{d_2} vp(st_j)}
          {d_1 + d_2}
   \right) \\
&\qquad
   - \frac{d_1}{d_2}
     \left(
       \frac{\sum_{j=1}^{d_1} vp(st_j)}{d_1}
     \right) \\
&= r(vr-vr_1).
\end{align*}
%So, we know that $vr_2-vr=r(vr-vr_1)$. 
The above further results in the following:
\begin{align*}
P_{p}^{st} - P_{r}^{st} &=(1-vr)^{n}-(1-vr_1)^{n_1}(1-vr_2)^{n_2}\\
&=(1-vr)^{n_1+n_2}-[1-vr+(vr-vr_1)]^{rn_2}\\
&\qquad\times[1-vr-r(vr-vr_1)]^{n_2}.
\end{align*}

Let: $(x=1-vr)$ and $(y=vr-vr_1)$, then we have:
\begin{align*}
P_{p}^{st} - P_{r}^{st}= x^{n_1+n_2}-(x+y)^{rn_2}(x-ry)^{n_2}.
\end{align*}

Suppose that $f(x,y)=x^{n_1+n_2}-(x+y)^{rn_2}(x-ry)^{n_2}$. Its partial derivative with respect to $y$ is:           
\begin{align*}
\frac{\partial f(x, y)}{\partial y} &=-r n_2(x+y)^{r n_2-1}(x-r y)^{n_2}\\
&\qquad+r n_2(x+y)^{r n_2}(x-r y)^{n_2-1}\\
&= r n_2 (x - r y)^{n_2 - 1} (x + y)^{r n_2 - 1} (1 + r) y.
\end{align*}
Because: 
(a)~$(x-ry)=(1-vr_2)>0$, 
(b)~$(x+y)=(1-vr_1)>0$, 
(c)~$r \geqslant 1$, and 
(d)~$n_2 > 0$,
the sign of $\frac{\partial f(x, y)}{\partial y}$ depends on $y$, and the minimal value of $f(x,y)$ is 0 at $y=0$. Therefore, we have:
\begin{align*}
P_{p}^{st} - P_{r}^{st} \geqslant 0. \\
\end{align*}

\noindent
The above proves that Proposition~\ref{prop:EffectivenessComparisonSt} holds when $k=2$.

Next, we assume that Proposition~\ref{prop:EffectivenessComparisonSt} also holds when $k \ge 2$, i.e.:
\begin{align*}
1 - 
\prod_{i=1}^k
\left(1 - \frac{\sum_{j=1}^{d_i} vp(st_j)}{d_i}\right)^{n_i}
&\geqslant
1 - \left(1 - \frac{\sum_{j=1}^{d} vp(st_j)}{d}\right)^{n}.
\end{align*}
Equivalently, we have
\begin{align*}
\prod_{i=1}^k
\left(1 - \frac{\sum_{j=1}^{d_i} vp(st_j)}{d_i}\right)^{n_i}
&\leqslant
\left(1 - \frac{\sum_{j=1}^{d} vp(st_j)}{d}\right)^{n}.
\end{align*}

Now, we consider the case with $k+1$ selection subdomains. We have:

\begin{align*}
P_{p}^{st} 
&= 1 -
   \prod_{i=1}^{k+1}
   \left(
     1 - \frac{\sum_{j=1}^{d_i} vp(st_j)}{d_i}
   \right)^{n_i} \\
&= 1 -
   \prod_{i=1}^{k}
   \left(
     1 - \frac{\sum_{j=1}^{d_i} vp(st_j)}{d_i}
   \right)^{n_i} 
\\ &\qquad\quad
   \times
   \left(
     1 - \frac{\sum_{j=1}^{d_{k+1}} vp(st_j)}{d_{k+1}}
   \right)^{n_{k+1}} \\
&\geqslant
   1 -
   \left(
     1 - \frac{\sum_{j=1}^{d'} vp(st_j)}{d'}
   \right)^{n'}
   \\&\qquad\quad
   \times
   \left(
     1 - \frac{\sum_{j=1}^{d_{k+1}} vp(st_j)}{d_{k+1}}
   \right)^{n_{k+1}},
\end{align*}

\noindent
where $n' = \sum_{i=1}^k n_i$ and $d' = \sum_{i=1}^k d_i$.

We can consider the union of the first $k$ selection subdomains as a single selection subdomain denoted by $D' = D^{st}_1 \cup D^{st}_2 \cup \cdots \cup D^{st}_k$. 
Assuming that $D'$ contains $d'$ test cases and that the number of source test cases to be selected from $D'$ is $n'$, the case with $k+1$ selection subdomains can be reduced to the case of two selection subdomains, namely $D'$ and $D^{st}_{k+1}$.
Here, the above proof for $k=2$ selection subdomains can be reused, thus, we have:
\begin{align*}
P_{p}^{st} &\geqslant 1-\left(1-\frac{\sum_{j=1}^{d'} vp\left(st_j\right)}{d'}\right)^{n'}\left(1-\frac{\sum_{j=1}^{d_{k+1}} vp\left(st_j\right)}{d_{k+1}}\right)^{n_{k+1}}\\
&\geqslant 1-\left(1-\frac{\sum_{j=1}^{d} vp\left(st_j\right)}{d}\right)^n \geqslant P_{r}^{st}.
\end{align*}

The above shows that PSALM is at least as effective as RS in terms of fault detection for selecting source test cases when there are $k+1$ selection subdomains. 
Thus, by mathematical induction, Proposition~\ref{prop:EffectivenessComparisonSt} holds. 
In summary, we can conclude that, regardless of the partition scheme used and the number of selection subdomains, the effectiveness of PSALM for source test case selection is at least on par with that of RS\@.
\end{proof}

\subsection{Using PSALM and RS for MG Selection}
\label{sec:pssVSrs4MGs}

This subsection theoretically compares the effectiveness of using PSALM and RS for MG selection.
Note that, unlike source test case selection, which is performed over $D^{st}$, MG selection is performed over $D^\mathrm{MG}$.

Given an MR, suppose there are $m$ defect-triggering MGs in $D^\mathrm{MG}$ and $m_i$ defect-triggering MGs in $D^\mathrm{MG}_i$ ($1 \le i \le k$).
The effectiveness of using PSALM and RS for MG selection is then expressed in Definitions~\ref{def:PPSSMG} and~\ref{def:PRTMG} below:

\begin{definition}[\bf Effectiveness of using PSALM for MG selection]
\label{def:PPSSMG}
The effectiveness of using PSALM for MG selection (denoted by $P_{p}^{\mathrm{MG}}$) is defined as:
\[
P_{p}^{\mathrm{MG}}
= 1 - \prod_{i=1}^k 
\left( 1 - \frac{m_i}{d_i^{\mathrm{MG}}} \right)^{n_i},
\]
where $k$ is the number of MG subdomains, 
$d_i^{\mathrm{MG}}$ is the size of $D_i^{\mathrm{MG}}$, 
and $n_i$ is the number of MGs selected from $D_i^{\mathrm{MG}}$.
\end{definition}

\begin{definition}[\bf Effectiveness of using RS for MG selection]
\label{def:PRTMG}
The effectiveness of using RS for MG selection (denoted by $P_{r}^{\mathrm{MG}}$) is defined as:
\[
  P_{r}^{\mathrm{MG}} = 1 - \left(1 -\frac{m}{d^{\mathrm{MG}}}\right)^n,
\]
where $n$ is the number of MGs randomly selected from $D^\mathrm{MG}$.
\end{definition}

With the above definitions, we have the following proposition:
\begin{proposition}[\bf Comparison between using PSALM and RS for MG selection]
\label{prop:EffectivenessComparisonMG}
Regardless of the partition scheme used, we have $P_{p}^{\mathrm{MG}} \geqslant P_{r}^{\mathrm{MG}}$.
\end{proposition}

% \newcontent{
% In traditional testing, the test unit is an individual test case; whereas in MT, the test unit becomes an MG\@.
% Note that, both test units can be classified as either failure-causing or non-failure-causing.
% Since the theoretical result of PSS in traditional testing does not depend on the specific selection domain and the specific partition scheme used, 
% the result and proof of PSS in traditional testing can be adapted to PSALM in MT by mapping the relevant concepts between traditional testing and MT\@. The mapping is shown in Table~\ref{tab:symbolsMapping}.
% Accordingly, Proposition~\ref{prop:EffectivenessComparisonMG} in the MT context follows immediately from the established result in traditional testing.
% }
\newcontent{
In traditional testing, the test unit is an individual test case; whereas in MT, the test unit becomes an MG\@.
Note that both test units can be classified as either failure-causing or non-failure-causing.
After the MT-specific concepts related to MG selection are formally defined,
the reasoning structure for analysing the effectiveness of PSALM and RS becomes analogous to that of classical PSS.
Therefore, rather than restating the full proof, we do not present the detailed derivation here and provide Table~\ref{tab:symbolsMapping} to illustrate the correspondence between the relevant concepts in traditional testing and MT.
Accordingly, Proposition~\ref{prop:EffectivenessComparisonMG} establishes the corresponding effectiveness relationship for MG selection under the MT-specific definitions.
}

\newcontent{
\begin{table}[!ht]
    % \footnotesize
    \centering
    \caption{Corresponding symbols used in the proof of Proposition~\ref{prop:EffectivenessComparisonMG}}
    \label{tab:symbolsMapping}
    \resizebox{\columnwidth}{!}{
        \input{tables/tab-symbolsMapping}
    }
\end{table}
}

%% file: tables/tab-symbols.tex
\begin{tabular}{cl}
\toprule
\textbf{Symbol} & \textbf{Meaning} \\ 
\midrule
$D^{st}$ & Source test case domain \\
$D^{st}_i$ & $i$th subdomain of $D^{st}$ \\
$D^{\mathrm{MG}}$ & Metamorphic group (MG) domain \\
$D^{\mathrm{MG}}_i$ & $i$th subdomain of $D^{\mathrm{MG}}$ \\
$d$ & Size of $D^{st}$ \\
$d_i$ & Size of $D^{st}_i$ \\
$d^{\mathrm{MG}}$ & Size of $D^{\mathrm{MG}}$ \\
$d^{\mathrm{MG}}_i$ & Size of $D^{\mathrm{MG}}_i$ \\
$n$ & Number of selected test units from selection domain \\
$n_i$ & Number of selected test units from subdomain \\
$k$ & Number of subdomains \\ 
\bottomrule
\end{tabular}

%% file: tables/tab-symbolsMapping.tex
\begin{tabular}{ll}
\toprule
\textsc{Traditional testing} & \textsc{Metamorphic testing} \\
\midrule
\textit{Test case} & \textit{MG} \\
\textit{Failure-causing test case} & \textit{Failure-causing MG} \\
\textit{Non-failure-causing test case} & \textit{Non-failure-causing MG} \\
\textit{Input domain} & \textit{MG domain} \\
\textit{Input subdomain} & \textit{MG subdomain} \\
\textit{Test case selection from input subdomain} & \textit{MG selection from MG subdomain} \\
\bottomrule
\end{tabular}

%% file: section/sec-phase1VSphase2.tex
\section{Theoretical Comparison on Using PSALM between Source Test Case Selection and MG Selection}
\label{sec:phase1VSphase2}

\newcontent{
In this section, we compare the fault detection effectiveness of applying PSALM to source test case selection and to MG selection.
}

%\newcontent{Since PSS can be applied to both source test case selection and MG selection in MT, and these two selections operate on fundamentally different selection domains, it is necessary to examine whether the resulting fault detection effectiveness is identical or exhibits systematic differences. This section provides such a theoretical comparison.

To establish the condition under which the two selections (one for source test cases and the other for MGs) using PSALM yield the same fault detection effectiveness, we first define the notion of equivalence between the partition scheme for the source test case domain and that for the MG domain.
\newcontent{Our aim is to induce a one-to-one correspondence between the source subdomains and the MG subdomains.}

\vspace{0.5ex} % adhoc
\begin{definition}[\bf Equivalence between two partition schemes]
\label{def:EquivalencePartition}
Given an MR, let:

\vspace{0.5ex} % adhoc
\begin{itemize}
  \item $\mathit{PS}_1$ be a partition scheme that partitions the source test case domain $D^{\mathrm{st}}$ into $k$ disjoint subdomains $\{D_1^{\mathrm{st}}, D_2^{\mathrm{st}}, \ldots, D_k^{\mathrm{st}}\}$; and
  \item $\mathit{PS}_2$ be a partition scheme that partitions the MG domain $D^{\mathrm{MG}}$ into $k$ disjoint subdomains $\{D_1^{\mathrm{MG}}, D_2^{\mathrm{MG}}, \ldots, D_k^{\mathrm{MG}}\}$.
\end{itemize}

\vspace{0.5ex} % adhoc
The two partition schemes $\mathit{PS}_1$ and $\mathit{PS}_2$ are considered \emph{equivalent} if both of the following conditions hold:

\vspace{0.5ex} % adhoc
\begin{itemize}
 \item For each $D_i^{\mathrm{st}}$ ($1 \leqslant i \leqslant k$), there exists an MG subdomain $D_j^{\mathrm{MG}}$ ($1 \leqslant j \leqslant k$) such that $\mathrm{MG}(D_i^{\mathrm{st}}) \subseteq D_j^{\mathrm{MG}}$,
 where $\mathrm{MG}(D_i^{\mathrm{st}})$ denotes the set of MGs constructed from all source test cases in $D_i^{\mathrm{st}}$.

 \item For each $D_j^{\mathrm{MG}}$ ($1 \leqslant j \leqslant k$), there exists a source subdomain $D_i^{\mathrm{st}}$ ($1 \leqslant i \leqslant k$) such that $\mathrm{ST}(D_j^{\mathrm{MG}}) \subseteq D_i^{\mathrm{st}}$,
 where $\mathrm{ST}(D_j^{\mathrm{MG}})$ denotes the set of source test cases associated with all MGs in $D_j^{\mathrm{MG}}$.
\end{itemize}
\end{definition}

\newcontent{To formalize the size relationship required for comparing the two selections using PSALM, we next establish the following lemma:

\begin{lemma}[\bf Size relationship between corresponding source test case subdomains and MG subdomains]
\label{lem:SubdomainSizeRelation}
Given an MR, if each source test case is associated with the same number of MGs, then for every pair of corresponding subdomains $(D_i^{st}, D_i^{MG})$ induced by the two equivalent partition schemes (one for $D^{st}$ and the other for $D^\mathrm{MG}$), there exists a constant $N$ such that:
\[
|D_i^\mathrm{{MG}}| = N \cdot |D_i^\mathrm{{st}}|.
\]
\end{lemma}

\begin{proof}
Consider any pair of corresponding subdomains $(D_i^\mathrm{{st}}, D_i^\mathrm{{MG}})$.
%%% under Definition~\ref{def:EquivalencePartition}.  
By equivalence, all MGs associated with the source test cases in $D_i^\mathrm{{st}}$ lie entirely within $D_i^\mathrm{{MG}}$, and all source test cases associated with the MGs in $D_i^\mathrm{{MG}}$ lie entirely within $D_i^\mathrm{{st}}$.  
Hence, the bipartite graph induced by the source test cases and the MGs is closed within this pair. 
This closure guarantees that all adjacency relations are internal to $(D_i^\mathrm{{st}}, D_i^\mathrm{{MG}})$, so that the total number of edges can be counted from either side of the bipartite graph.

Given any given MR, each associated MG contains the same number of source test cases (determined solely by the input structure of the MR), and each source test case  is associated with the same number of MGs. 

Let $r$ denote the number of source test cases associated with each MG, and $L$ denote the number of MGs associated with each source test case.  
Also, let $E_i$ be the number of edges in the induced bipartite graph on $(D_i^\mathrm{{st}}, D_i^\mathrm{{MG}})$.

Counting edges from the source-test-case side of the graph yields:
\[
E_i = |D_i^\mathrm{{st}}| \cdot L;
\]
whereas counting edges from the MG side of the graph yields:
\[
E_i = |D_i^\mathrm{{MG}}| \cdot r.
\]
Now, equating the above two expressions gives:
\[
|D_i^\mathrm{{MG}}| = \frac{L}{r} \cdot |D_i^\mathrm{{st}}|.
\]
In the above equation, we will not encounter the case of division by zero, since $r$ denotes the number of source test cases associated with each MG and, hence, $r > 0$. Finally, letting $N = L/r$ completes the proof.
\end{proof}

Lemma~\ref{lem:SubdomainSizeRelation} above enables us to characterize the conditions under which using PSALM for selecting source test cases and selecting MGs achieve identical fault-detection effectiveness.
}

\newcontent{
\begin{proposition}[\bf Effectiveness equality between source test case selection and MG selection under the same partition scheme]
\label{prop:EffectivenessEquality}
Given an MR, if: (a)~the partition scheme for the source test case domain is equivalent to that for the MG domain, 
and (b)~each source test case is associated with the same number of MGs, 
then $P_p^{\mathrm{st}} = P_p^{\mathrm{MG}}$.
\end{proposition}

\begin{proof}
By Lemma~\ref{lem:SubdomainSizeRelation}, the equivalence between the two partition schemes and the fact that each source test case is associated with the same number of MGs imply that, for every pair of corresponding subdomains:
\[
|D_i^{\mathrm{MG}}| = N \cdot |D_i^{\mathrm{st}}|
\quad\text{and}\quad
d_i^{\mathrm{MG}} = N \cdot d_i^{\mathrm{st}},
\]
where $N$ is a constant determined by the MR\@.

Recall that under PSALM, the number of selected units in each subdomain is proportional to the size of that subdomain.  
Let $n$ be the total number of selected units. Then, we have:
\[
n_i^{\mathrm{st}} = n \cdot \frac{|D_i^{\mathrm{st}}|}{|D^{\mathrm{st}}|}
\quad\text{and}\quad
n_i^{\mathrm{MG}} = n \cdot \frac{|D_i^{\mathrm{MG}}|}{|D^{\mathrm{MG}}|}.
\]
Using $|D_i^{\mathrm{MG}}| = N|D_i^{\mathrm{st}}|$ and $|D^{\mathrm{MG}}| = N|D^{\mathrm{st}}|$ (see Lemma~\ref{lem:SubdomainSizeRelation}) yields:
\[
n_i^{\mathrm{MG}} = n_i^{\mathrm{st}}
\quad\text{for all } i.
\]

The effectiveness expressions in Definitions~\ref{def:PPss} and~\ref{def:PPSSMG} are therefore:
\begin{align*}
P_p^{\mathrm{st}} 
= 1 - \prod_{i=1}^{k} 
\left(1 - \frac{\sum_{j=1}^{d_i^{\mathrm{st}}} vp(st_j)}{d_i^{\mathrm{st}}}\right)^{n_i^{\mathrm{st}}}\\
P_p^{\mathrm{MG}}
= 1 - \prod_{i=1}^{k}
\left(1 - \frac{m_i}{d_i^{\mathrm{MG}}}\right)^{n_i^{\mathrm{MG}}}.
\end{align*}

For any source test case $st_j$, let $\mathrm{MG}_{st_j}$ and $\mathrm{MG}_{st_j}^v$ denote the sets of associated MGs and failure-causing MGs, respectively.
Definition~\ref{def:stEffectiveness} gives:
\[
vp(st_j) = \frac{|\mathrm{MG}_{st_j}^v|}{|\mathrm{MG}_{st_j}|}.
\]
Since each source test case is associated with exactly $N$ MGs with respect to the relevant MR, we have:
\[
vp(st_j) = \frac{|\mathrm{MG}_{st_j}^v|}{N}.
\]
In turn, we have:
\[
\frac{\sum_{j=1}^{d_i^{\mathrm{st}}} vp(st_j)}{d_i^{\mathrm{st}}}
= \frac{\sum_{j=1}^{d_i^{\mathrm{st}}} |\mathrm{MG}_{st_j}^v|}{N d_i^{\mathrm{st}}}.
\]

By Definition~\ref{def:EquivalencePartition}, 
the MGs in $D_i^{\mathrm{MG}}$ are exactly those that are associated with source test cases in $D_i^{\mathrm{st}}$.
Hence:
\[
m_i = \sum_{j=1}^{d_i^{\mathrm{st}}} |\mathrm{MG}_{st_j}^v|,
\quad\text{and}\quad
d_i^{\mathrm{MG}} = N d_i^{\mathrm{st}},
\]
which together yield:
\[
\frac{\sum_{j=1}^{d_i^{\mathrm{st}}} vp(st_j)}{d_i^{\mathrm{st}}}
= \frac{m_i}{d_i^{\mathrm{MG}}}.
\]
Combining:
\[
n_i^{\mathrm{MG}} = n_i^{\mathrm{st}}
\quad\text{and}\quad
\frac{\sum vp(st_j)}{d_i^{\mathrm{st}}} = \frac{m_i}{d_i^{\mathrm{MG}}},
\]
we can see that the two effectiveness formulae match term by term. Therefore, we have:
\[
P_p^{\mathrm{st}} = P_p^{\mathrm{MG}}.
\]
\end{proof}

}

\newcontent{

Recall that Proposition~\ref{prop:EffectivenessEquality} applies when:
(a)~each source test case is associated with the same number of MGs, and 
(b)~the partition scheme for the source test case domain is equivalent to that for the MG domain.
Now, we introduce two more propositions below.
Propositions~\ref{prop:condition1} and~\ref{prop:condition2} deal with the scenario where conditions~(a) and~(b) above are not met, respectively.

\begin{proposition}[\bf Effectiveness relationship between source test case selection and MG selection when source tests are associated with different number of MGs]
\label{prop:condition1}
Given an MR, if source test cases are associated with different numbers of MGs, the effectiveness of applying PSALM to source test case selection ($P_{p}^\mathrm{st}$) may be greater than, equal to, or less than that of applying PSALM to MG selection ($P_{p}^\mathrm{MG}$).
%\footnote{\zzh{I think we do not need to emphasize that the partition scheme for the source test case domain is equivalent to that for the MG domain. As long as different source test cases correspond to different numbers of MGs, the proposition holds.}}
\end{proposition}
}

\begin{proof}
We use an example to validate Proposition~\ref{prop:condition1}.
In this example, without loss of generality, given an MR, suppose that this MR is associated with multiple MGs, and each such MG contains one source test case and one follow-up test case. 
%This assumption is adopted for the purpose of constructing the example and does not change the general definition of MGs used elsewhere in this paper.

Suppose further that, with respect to a given MR, we have three source test cases: $\mathit{st}_1$, $\mathit{st}_2$, and $\mathit{st}_3$.  
Each of these source test cases generates a different number of follow-up test cases and, hence, is associated with a different numbers of MGs.  
Specifically, $\mathit{st}_1$ generates one follow-up test case and is therefore associated with one MG (denoted by $\mathrm{MG}_{11}$); 
$\mathit{st}_2$ generates two follow-up test cases and is associated with two MGs (denoted by $\mathrm{MG}_{21}$ and $\mathrm{MG}_{22}$); 
and $\mathit{st}_3$ generates three follow-up test cases and is associated with three MGs (denoted by $\mathrm{MG}_{31}$, $\mathrm{MG}_{32}$, and $\mathrm{MG}_{33}$).

We now examine the effectiveness of PSALM in source test case selection and MG selection under the above scenario.
Table~\ref{tab:example4proposition4} presents three cases.
In this table, defect-triggering MGs are annotated with the superscript $v$.

We first compute the effectiveness of PSALM in source test case selection. 
In Case~1, the source test case domain is divided into two subdomains: 
$\{\mathit{st}_2\}$ and $\{\mathit{st}_1, \mathit{st}_3\}$.
For $\mathit{st}_1$, since it generates only one MG (MG$_{11}$) and this MG is non-defect-triggering, we have $\mathit{vp}(\mathit{st}_1)=0$ (Definition~\ref{def:stEffectiveness}).  
For $\mathit{st}_2$, both associated MGs (MG$_{21}$ and MG$_{22}$) are non-defect-triggering, so $\mathit{vp}(\mathit{st}_2)=0$.  
For $\mathit{st}_3$, all three associated MGs (MG$^v_{31}$, MG$^v_{32}$, MG$^v_{33}$) are defect-triggering, so $\mathit{vp}(\mathit{st}_3)=1$.
By Definition~\ref{def:PPss}, we have:  
\[
P^{\mathit{st}}_p
= 1 - (1 - 0)\bigl(1 - (0 + 1)/2\bigr)
= 0.5.
\]

We then compute the effectiveness of PSALM in MG selection.
The MG domain is partitioned into two subdomains:
$\{\mathrm{MG}_{21}, \mathrm{MG}_{22}\}$ and
$\{\mathrm{MG}_{11}, \mathrm{MG}_{31}^v, \mathrm{MG}_{32}^v, \mathrm{MG}_{33}^v\}$.
For the first MG subdomain, $m_1 = 0$ and $d_1 = 2$.
For the second MG subdomain, there is one non-defect-triggering MG and three defect-triggering MGs, so $m_2=3$ and $d_2=4$.
By Definition~\ref{def:PPSSMG}, we have:
\[
P^{\mathrm{MG}}_p
= 1 - (1 - 0)(1 - 3/4)
= 0.75.
\]
Thus, for Case~1, we have $P^{\mathit{st}}_p < P^{\mathrm{MG}}_p$.

Applying the same computation to Case~2 and Case~3 in  
Table~\ref{tab:example4proposition4}, we obtain  
$P^{\mathit{st}}_p = P^{\mathrm{MG}}_p$ and  
$P^{\mathit{st}}_p > P^{\mathrm{MG}}_p$, respectively.
These three cases collectively show that, when source test cases are associated with different numbers of MGs, the effectiveness of PSALM in source test case selection may be greater than, equal to, or less than in MG selection, thereby validating  Proposition~\ref{prop:condition1}.
\end{proof}

\begin{table*}[ht]
    \centering
    \caption{Three cases for validating Proposition~\ref{prop:condition1}}
    \label{tab:example4proposition4}
    \resizebox{!}{!}{
        \input{tables/tab-example4proposition4}
    }
\end{table*}

\newcontent{\begin{proposition}[\bf Effectiveness relationship between source test case selection and MG selection under non-equivalent partition schemes]
\label{prop:condition2}
If the partition schemes for the source test case domain and the MG domain are non-equivalent, the effectiveness of applying PSALM to source test case selection ($P_{p}^{\mathrm{st}}$) may be greater than, equal to, or less than that of applying PSALM to MG selection ($P_{p}^{\mathrm{MG}}$).
\end{proposition}
}

\begin{proof}
To validate Proposition~\ref{prop:condition2}, we use an example in which it is assumed that each MG contains one source test case and one follow-up test case (as in the preceding example used for validating Proposition~\ref{prop:condition1}). 

Suppose, with respect to a given MR, we have three source test cases: $\mathit{st}_1$, $\mathit{st}_2$, and $\mathit{st}_3$.  
Suppose, further that, each of these source test cases generates two follow-up test cases and is therefore associated with two MGs. More specifically, $\mathit{st}_1$ is associated with $\mathrm{MG}_{11}$ and $\mathrm{MG}_{12}$. 
Similarly, $\mathit{st}_2$ is associated with $\mathrm{MG}_{21}$ and $\mathrm{MG}_{22}$; and $\mathit{st}_3$ is associated with $\mathrm{MG}_{31}$ and $\mathrm{MG}_{32}$.

We now compare the effectiveness of PSALM for source test case selection and under MG selection.  
Table~\ref{tab:example4proposition5} gives three cases involving using non-equivalent partition schemes.  
In this table, defect-triggering MGs are annotated with the superscript~$v$.

We look at Case~1 first.  
The source test case domain is partitioned into two subdomains:  
$\{\mathit{st}_1, \mathit{st}_2\}$ and $\{\mathit{st}_3\}$.  
Similarly, the MG domain is partitioned into two subdomains:
$\{\mathrm{MG}_{11}^v, \mathrm{MG}_{21}^v\}$ and  
$\{\mathrm{MG}_{12}, \mathrm{MG}_{22}, \mathrm{MG}_{31}, \mathrm{MG}_{32}\}$.

For source test case selection,  
$\mathit{vp}(\mathit{st}_1) = 0.5$ because one of its two MGs is defect-triggering;  
$\mathit{vp}(\mathit{st}_2)=0.5$ for the same reason;  
and $\mathit{vp}(\mathit{st}_3)=0$ (since both MGs associated with $\mathit{st}_3$ are non-defect-triggering).  
By Definition~\ref{def:PPss}, we have:
\[
P_{p}^{\mathrm{st}}
= 1 - \bigl( 1 - (0.5+0.5)/2 \bigr)(1 - 0)
= 0.5.
\]

For MG selection,  
the first MG subdomain contains two defect-triggering MGs, so $m_1=2$ and $d_1=2$;  
the second subdomain contains only non-defect-triggering MGs, so $m_2=0$.  
By Definition~\ref{def:PPSSMG}:
\[
P_{p}^{\mathrm{MG}}
= 1 - (1 - 1)(1 - 0)
= 1.
\]
Thus $P_{p}^{\mathrm{st}} < P_{p}^{\mathrm{MG}}$ for Case~1.

Using the same computation method,  
Case~2 yields $P_{p}^{\mathrm{st}} = P_{p}^{\mathrm{MG}}$,  
and Case~3 yields $P_{p}^{\mathrm{st}} > P_{p}^{\mathrm{MG}}$.
These results collectively validate Proposition~\ref{prop:condition2}.
\end{proof}

\begin{table*}[!ht]
    % \footnotesize
    \centering
    \caption{Three cases for validating Proposition~\ref{prop:condition2}}
    \label{tab:example4proposition5}
    \resizebox{!}{!}{
        \input{tables/tab-example4proposition5}
    }
\end{table*}

\newcontent{
In summary, Propositions~\ref{prop:EffectivenessEquality}, \ref{prop:condition1}, and~\ref{prop:condition2} together suggest that, when the partition scheme for the source test case domain is equivalent to that for the MG domain, and each source test case is associated with the same number of MGs, then we have $P_p^{st} = P_p^\mathrm{MG}$. Otherwise, we cannot theoretically determine whether $P_p^{st}$ is larger than, equal to, or less than $P_p^\mathrm{MG}$.
}

%% file: tables/tab-example4proposition4.tex
\begin{tabular}{cccccc}
  \toprule
  \textbf{Case} &
  \textbf{Source test case selection} &
  \textbf{MG selection} &
  \textbf{$P_{p}^{\mathrm{st}}$} &
  \textbf{$P_{p}^{\mathrm{MG}}$} &
  \textbf{Remarks} \\
  \midrule

  1 &
  \begin{tabular}[c]{@{}c@{}}\{$st_2$\}\\ \{$st_1, st_3$\}\end{tabular} &
  \begin{tabular}[c]{@{}c@{}}\{$\mathrm{MG}_{21}, \mathrm{MG}_{22}$\}\\
                               \{$\mathrm{MG}_{11}, \mathrm{MG}_{31}^v, \mathrm{MG}_{32}^v, \mathrm{MG}_{33}^v$\}\end{tabular} &
  0.50 &
  0.75 &
  $P_{p}^{\mathrm{st}} < P_{p}^{\mathrm{MG}}$ \\
  \midrule

  2 &
  \begin{tabular}[c]{@{}c@{}}\{$st_2$\}\\ \{$st_1, st_3$\}\end{tabular} &
  \begin{tabular}[c]{@{}c@{}}\{$\mathrm{MG}_{21}^v, \mathrm{MG}_{22}^v$\}\\
                               \{$\mathrm{MG}_{11}, \mathrm{MG}_{31}, \mathrm{MG}_{32}, \mathrm{MG}_{33}$\}\end{tabular} &
  1.00 &
  1.00 &
  $P_{p}^{\mathrm{st}} = P_{p}^{\mathrm{MG}}$ \\
  \midrule

  3 &
  \begin{tabular}[c]{@{}c@{}}\{$st_2$\}\\ \{$st_1, st_3$\}\end{tabular} &
  \begin{tabular}[c]{@{}c@{}}\{$\mathrm{MG}_{21}, \mathrm{MG}_{22}$\}\\
                               \{$\mathrm{MG}_{11}^v, \mathrm{MG}_{31}, \mathrm{MG}_{32}, \mathrm{MG}_{33}$\}\end{tabular} &
  0.50 &
  0.25 &
  $P_{p}^{\mathrm{st}} > P_{p}^{\mathrm{MG}}$ \\
  \bottomrule
\end{tabular}

%% file: tables/tab-example4proposition5.tex
    \begin{tabular}{cccccc}
        \toprule
        \textbf{Case} &
        \textbf{Source test case selection} &
        \textbf{MG selection} &
        \textbf{$P_{p}^{\mathrm{st}}$} &
        \textbf{$P_{p}^{\mathrm{MG}}$} &
        \textbf{Remarks} \\
        \midrule

        1 &
        \begin{tabular}[c]{@{}c@{}}\{$st_1, st_2$\}\\ \{$st_3$\}\end{tabular} &
        \begin{tabular}[c]{@{}c@{}}\{$\mathrm{MG}_{11}^v,\mathrm{MG}_{21}^v$\}\\
                                   \{$\mathrm{MG}_{12},\mathrm{MG}_{22},\mathrm{MG}_{31},\mathrm{MG}_{32}$\}\end{tabular} &
        0.50 & 1.00 &
        $P_{p}^{\mathrm{st}} < P_{p}^{\mathrm{MG}}$ \\
        \midrule

        2 &
        \begin{tabular}[c]{@{}c@{}}\{$st_2$\}\\ \{$st_1, st_3$\}\end{tabular} &
        \begin{tabular}[c]{@{}c@{}}\{$\mathrm{MG}_{11}^v, \mathrm{MG}_{12}$\}\\
                                   \{$\mathrm{MG}_{21}^v, \mathrm{MG}_{22},\mathrm{MG}_{31},\mathrm{MG}_{32}$\}\end{tabular} &
        0.63 & 0.63 &
        $P_{p}^{\mathrm{st}} = P_{p}^{\mathrm{MG}}$ \\
        \midrule

        3 &
        \begin{tabular}[c]{@{}c@{}}\{$st_1$\}\\ \{$st_2, st_3$\}\end{tabular} &
        \begin{tabular}[c]{@{}c@{}}\{$\mathrm{MG}_{12}$\}\\
                                   \{$\mathrm{MG}_{11}^v, \mathrm{MG}_{21}^v, \mathrm{MG}_{22},\mathrm{MG}_{31},\mathrm{MG}_{32}$\}\end{tabular} &
        0.63 & 0.40 &
        $P_{p}^{\mathrm{st}} > P_{p}^{\mathrm{MG}}$ \\
        \bottomrule
    \end{tabular}

%% file: section/sec-experimentSetup.tex
\section{Research Questions and Experimental Setup}
\label{sec:experimentSetup}
To evaluate the effectiveness and efficiency of PSALM, we designed and performed a series of experiments. 
This section first introduces the research questions that guide our empirical study. Then, we discuss the subject programs, comparison baselines, evaluation metrics, configurations, and the replication package.

%============================================================================================================================================%

% Our experiment was conducted on a computer with the configuration shown in Table~\ref{tab:computer}:

% \begin{table}[!ht]
%     % \footnotesize
%     \centering  
%     % \renewcommand{\arraystretch}{0.7}  
%     \caption{Computing Configuration}
%     \label{tab:computer}
%     \resizebox{\columnwidth}{!}{
%         \input{tables/tab:computer}
%     }
% \end{table}
\subsection{Research Questions}
\label{sec:rq}

\newcontent{To systematically evaluate the effectiveness and limitations of PSALM, we performed our empirical study around the following three research questions:

\vspace{0.5ex} % adhoc
\begin{itemize}
  \item \textbf{RQ1:} \RQContent{1}
  \item \textbf{RQ2:} \RQContent{2}
  \item \textbf{RQ3:} \RQContent{3}
  % \item \textbf{RQ4:} \RQContent{4}
\end{itemize}

\vspace{0.5ex} % adhoc
RQ1 investigates whether PSALM achieves its fundamental goal of improving fault-detection effectiveness over RS, which is argued to be the most widely used baseline for benchmarking~\cite{BarusChen2016, SunLiu2022}. 
RQ2 investigates the relative effectiveness of PSALM in two important selections of MT (see Steps~2 and~4 of MT in Section~\ref{sec:intro}). 
Finally, RQ3 investigates the effectiveness of PSALM in the broader context by comparing it with two other representative selection strategies (ART and MT-ART), thereby highlighting its relative advantages and limitations. 
Thus, these three research questions together provide a comprehensive evaluation of PSALM\@. 
}

\subsection{Subject Programs}
\label{sec:subjectPrograms}

\newcontent{
We evaluated PSALM using two groups of subject programs.
The first group consists of three benchmark programs used in prior work~\cite{ChenYu1994}, all implemented in Python (see Table~\ref{tab:subjectPrograms}).
To capture diverse partitioning characteristics, these programs were selected to represent three typical scenarios:
(i) subdomains with non-integer size ratios,
(ii) subdomains whose sizes differ by several orders of magnitude, and
(iii) input domains that induce an unbounded number of subdomains, which were merged into a finite set for feasibility.
}

\newcontent{
The second group comprises five methods extracted from Defects4J projects~\cite{JustJalali2014} (see Table~\ref{tab:subjectPrograms}).
For the rest of this paper, we refer to the programs in Group~1 and the methods in Group~2 collectively as ``programs''.
Defects4J is a widely used benchmark of real-world Java programs.
The selected programs cover diverse application domains, including date and time processing, string parsing, mathematical utilities, and graphical libraries.
Compared with the programs in the first group, these programs are larger and more complex, allowing us to assess the practicality and scalability of PSALM in realistic software systems.
}

\newcontent{
Mutants were generated for the programs in the first group using \textsc{Mutmut}~\cite{hovmoller2025}, and for the programs in the second group using \textsc{Major}~\cite{Just2014}.
In both cases, we first removed equivalent and trivial mutants.
We then applied subsumption analysis to eliminate subsumed mutants, retaining only non-subsumed mutants.
This was to reduce validity threats arising from the inclusion of redundant mutants.
}

\newcontent{
Thereafter, we defined sets of MRs to serve as the test oracles. For each subject program, we manually defined MRs based on the inherent properties of the target functionality. The number of MRs for each program is summarized in Table~\ref{tab:subjectPrograms}.
This table gives information about each subject program, including its implementation language, lines of code (LOC), the number of retained mutants, the number of defined MRs, and functionality.
% \footnote{\tt I deleted the phrase ``no. of partitions'' because this information is not provided in the table. (PL)}
The complete set of subject programs and MRs is provided in our replication package.
\footnote{\, To facilitate transparency and reproducibility, we provide an artifact
containing the experimental code and scripts used in this study.
The artifact is publicly available at \url{https://doi.org/10.5281/zenodo.17926177}.}
%%% Note that the reported LOC was calculated following the practice in prior work~\cite{SunLiu2022}.}
}

% \newcontent{
% For \inlineCode{IncomeTax}, which computes the tax based on income, we defined a monotonicity MR stating that the computed tax must not decrease when the income increases. Formally,
% \begin{equation*}
% IT(x) \leq IT(x+\delta), \quad \delta > 0,
% \end{equation*}
% where $IT(x)$ denotes the tax for income $x$.
% }

\begin{table*}[!ht]
    \ncolor
    % \footnotesize
    \centering
    \caption{Subject programs/methods used in our empirical study}
    \label{tab:subjectPrograms}
    % \resizebox{\linewidth}{!}{
        \input{tables/tab-subjectPrograms}
    % }
\end{table*}

% \newcontent{
% For \inlineCode{IncomeTax}, which computes the tax based on income, we defined a monotonicity MR stating that the computed tax must not decrease when the income increases. Formally,
% \begin{equation*}
% IT(x) \leq IT(x+\delta), \quad \delta > 0,
% \end{equation*}
% where $IT(x)$ denotes the tax for income $x$.
% }

% \newcontent{
% As another example from Defects4J, we defined a scaling MR for \inlineCode{parseInt}. This MR states that appending a digit ‘0’ to a decimal string should yield a value equal to ten times the parsed integer.
% \begin{equation*}
% \texttt{parseInt("1230")} = 10 \times \texttt{parseInt("123")}.
% \end{equation*}
% }

\subsection{Baseline selection strategies}
\label{sec:baselines}
\newcontent{
We selected baseline selection strategies according to the following two criteria: 
(a)~these strategies have been used in previous MT studies, and 
(b)~these strategies do not rely on execution feedback information. 
Criterion~(b) ensures that any observed difference in MT performance between PSALM and a baseline selection strategy is due to the intrinsic mechanism of the strategy itself, rather than advantages gained from additional runtime information. 
With criteria~(a) and~(b), we selected three baseline strategies: Random Selection, Adaptive Random Testing (ART), and MT-ART.}

\newcontent{
\textbf{Random Selection (RS).} RS is the most widely-used selection strategy in existing MT studies~\cite{BarusChen2016,SunLiu2022} and, thus, naturally inclines itself to be a baseline strategy. Moreover, our theoretical analysis in Section~\ref{sec:pssVSrs} demonstrates that PSALM is never inferior to RS in terms of fault detection effectiveness. Therefore, using RS as a baseline strategy allows us to validate our theoretical analysis in practice.}

\newcontent{
\textbf{Adaptive Random Testing (ART).} ART is a classical improvement over random testing and has also been applied in the MT context~\cite{BarusChen2016}. ART aims to enhance the effectiveness of RS by distributing test cases more evenly across the input space. Since PSALM adopts the concept of proportional sampling (which is related to test case distribution), comparing PSALM with ART allows us to evaluate the relative advantages between these two test case distribution strategies in the MT context.}

\newcontent{
\textbf{MT-ART.} MT-ART is a technique proposed by Hui et al.~\cite{HuiWang2021}, which generates MGs using distance metrics to improve fault detection effectiveness. Their results show that, besides MRs, selecting source and follow-up test cases (i.e., MGs) also plays a significant part in the effectiveness of MT~\cite{HuiWang2021}. Thus, MT-ART is a representative baseline strategy for MG selection.}

\newcontent{
Another reason for selecting ART and MT-ART as our baseline strategies was that both techniques correspond to the two selection steps of PSALM\@. 
ART, originally designed as a general test case generation strategy, can be applied to source test case selection, whereas MT-ART was specifically developed for MT and focuses on MG selection (where both source and follow-up test cases are selected together as MGs). 
Choosing ART and MT-ART enables us to benchmark PSALM with representative strategies for both source test case selection and MG selection.}

\subsection{Evaluation Metrics}
\label{sec:metrics}
\newcontent{
We used the following metrics to assess the effectiveness and efficiency of PSALM and the baseline strategies:}

\vspace{0.05ex} % adhoc
\newcontent{
\textbf{\textit{P}-measure.} This metric quantifies the probability that the selected test cases or MGs reveal a fault. Higher $P$-measure values thus indicate stronger fault detection effectiveness. We conducted significance testing using the two-tailed t-test with a confidence level of 0.05\@. 
}

\newcontent{
\textbf{Vargha--Delaney A12.}
This non-parametric effect-size statistic quantifies the probability that PSALM achieves a higher $P$-measure than a baseline strategy.
In our setting, an A12 value greater than 0.5 indicates that PSALM tends to outperform the baseline, whereas a value below 0.5 indicates the opposite.
An A12 value of 0.5 indicates no performance difference between the two strategies.
}

\subsection{Configurations}
\label{sec:configurations}
\newcontent{
Below we discuss our experiment configurations to address RQ1 to RQ3:}

\vspace{0.5ex} % adhoc
\newcontent{
\textbf{Selection sizes.}
In all our experiments, we set the number of selected source test cases and MGs to be three times the number of subdomains of the corresponding subject program.
This configuration allowed sufficient coverage of each subdomain while avoiding excessive test cases that would make our experiments unmanageable.
% would make the P-measure approach one and thus reduce the discriminative power between different strategies.
% As this study primarily aims to evaluate the effectiveness of our proposed method, the choice of test case numbers is not within the scope of this work.
}

\newcontent{
\textbf{Repetitions.} Each setting was independently repeated 30 times.
This follows common practice in software testing experiments, where 30 repetitions provide sufficiently stable estimates of the $P$-measure while mitigating the impact of randomness. 
In each trial, 1{,}000 iterations were performed to compute the $P$-measure, which further reduces variance and improves the reliability of the statistical results.
}

\newcontent{
\textbf{Environment.} All experiments were conducted on a workstation equipped with an Intel Xeon Silver~4210R CPU @\,2.40\,GHz, an NVIDIA RTX~3090 GPU (24\,GB), and 128\,GB memory, with Ubuntu~22.04.4 LTS as the operating system. The three Python benchmark programs in Group~1 were implemented and executed with Python~3.9.7, and the five Java programs in Group~2 from Defects4J were compiled and executed with Java JDK~21 in a Maven-managed environment.
}

%% file: tables/tab-subjectPrograms.tex
% \begin{tabular}{llccccl}
% \hline
% Program / Method & Language & LOC & \#Partitions & \#Mutants & \#MRs & Description \\ \hline
% \inlineCode{MortgageRate} & Python & 17  & 2 & 6  & 6 & Computes the mortgage rate \\
% \inlineCode{GeometricSum} & Python & 18  & 7 & 8  & 3 & Computes the geometric sum \\ 
% \inlineCode{IncomeTax}    & Python & 27  & 6 & 17 & 3 & Computes the income tax \\ \hline
% \inlineCode{copySign}        & Java & 38  & 4 & 35 & 5 & Returns the first argument with the sign of the second\\ 
% \inlineCode{createLineRegion} & Java & 54  & 4 & 82 & 6 & Creates a line region\\
% \inlineCode{parseInt}        & Java & 59  & 6 & 45 & 5 & Parses a string into an integer value\\ 
% \inlineCode{isSameDay}       & Java & 60  & 9 & 6  & 5 & Checks whether two dates fall on the same day\\ 
% \inlineCode{convolve}        & Java & 66  & 12 & 14 & 5 & Computes the convolution of two one-dimensional double arrays\\  \hline
% \end{tabular}

\begin{tabularx}{\linewidth}{lccccX}
\toprule
Program/method & Language & LOC & \#\,Mutants & \#\,MRs & Functionality \\
\midrule
MortgageRate (MoR) & Python & 17 & 4 & 6 & Computes the mortgage rate \\
IncomeTax (InT) & Python & 27 & 18 & 3 & Computes the income tax \\ 
GeometricSum (GeS) & Python & 18 & 4 & 3 & Computes the geometric sum \\ \midrule
MathUtils.copySign (CopyS) & Java & 269 & 35 & 5 & Returns the first argument with the sign of the second \\
NumberInput.parseInt (ParseI) & Java & 319 & 45 & 5 & Parses a string into an integer value \\
ShapeUtilities.createLineRegion (CLine) & Java & 709 & 58 & 6 & Creates a line region \\
MathArrays.convolve (Conv) & Java & 1537 & 14 & 5 & Computes the convolution of two arrays \\
DateUtils.isSameDay (IsDay) & Java & 1888 & 6 & 5 & Checks whether two dates fall on the same day \\
\bottomrule
\label{tab:subjectprograms}
\end{tabularx}

%% file: section/sec-results.tex
\section{Experimental Results and Analysis}
\label{sec:results}

\newcontent{This section reports the results of our empirical study and answers the three research questions in Section~\ref{sec:rq}. 
Each subsection presents the results related to one research question, followed by the related analysis.}

\subsection{RQ1: Performance comparison between PSALM and RS}
\newcontent{
RQ1 investigates whether or not PSALM consistently achieves equal or higher fault detection effectiveness than RS in practice.
We report the mean $P$-measures for the 30 trials (to avoid verbosity, for the rest of this paper, mean $P$-measure will simply be referred to as ``$P$-measure'')
of using PSALM and RS across all subject programs, and analyze their differences in terms of statistical significance and effect size.
Tables~\ref{tab:rq1phase1} and~\ref{tab:rq1phase2} give the $P$-measures across all mutants, together with the percentage improvements\,%
\footnote{\,In Tables~\ref{tab:rq1phase1}, \ref{tab:rq1phase2}, \ref{tab:rq2}, \ref{tab:rq3phase1}, and~\ref{tab:rq3phase2},
the 4th column ``Improvement (\%)'' shows the percentage improvement of the $P$-measure in the 2nd column over the $P$-measure in the 3rd column.},
$p$-values, and A12 statistics, for test case selection and MG selection, respectively.
}

\newcontent{
We first look at Table~\ref{tab:rq1phase1} 
which shows the performance of PSALM and RS in source test case selection.
Overall, PSALM outperforms RS (A12 $\geqslant 0.970$; $p < 0.05$) in six (MoR, GeS, InT, CopyS, CLine, and IsDay) out of the eight subject programs.
Also, the percentage improvements are substantial for InT (236.27\%), CLine (71.1\%), and IsDay (341.1\%).
Only ParseI and Conv show very small negative percentage improvements ($-$0.01\% or $-$2.36\%).
Next, Table~\ref{tab:rq1phase2} shows the performance of PSALM and RS in MG selection. 
Overall, PSALM achieves much better performance in seven subject programs (the top seven in Table~\ref{tab:rq1phase2}), with very large A12 ($\geqslant 0.967$) and $p < 0.05$.
Among these seven programs, InT shows the largest percentage improvement (436.14\%).
% The only exception is related to ISDay, showing no significant difference in performance between PSALM and RS ($p$-value $= 0.309$, which is larger than $0.05$).

Notably, the improvement of PSALM over RS varies between source test case selection and MG selection.
For example, 
% with respect to \inlineCode{ParseI} and \inlineCode{Conv}, the difference in performance between PSALM and RS for source test case selection is negligible, 
% but PSALM shows much better performance (in terms of A12 values; $\geqslant 0.998$) for MG selection.
with respect to IsDay, PSALM shows a large improvement over RS in source test case selection (percentage improvement = 341.13\%; A12 $=$ 1.000), yet exhibits comparable performance to RS in MG selection (percentage improvement $=$ 1.28\%; A12 $=$ 0.598). Such variation may be caused by changes in the selection domain ($D^{st}$ or $D^\mathrm{MG}$) and the distribution of failure-causing test units (source test cases or MGs) in the selection domain.

%We take a close examination of using PSALM for source test case selection for \inlineCode{Conv}, where the A12 value $=0.000$, 
%indicating that PSALM performs consistently worse than RS over all the 30 runs.
%We speculate that this ``surprising'' result is mainly caused by a few mutants with particularly low $P$-measure values under PSALM, thereby largely reducing the overall A12 value across all mutants for \inlineCode{Conv}.
We further analyze the performance between PSALM and RS at the mutant level. Figure~\ref{fig:rq1phase1} shows the distribution of mutants among the three outcome categories: PSALM better, no difference, and RS better, with respect to source test case selection. Figure~\ref{fig:rq1phase2} shows a similar mutant distribution, but this time for MG selection.
}

\begin{table}[!ht]
    \ncolor
    % \footnotesize
    \centering
    \caption{Comparison between PSALM and RS in source test case selection}
    \label{tab:rq1phase1}
    \resizebox{\columnwidth}{!}{
        \input{tables/tab-rq1phase1}
    }
\end{table}

\begin{table}[!ht]
    \ncolor
    % \footnotesize
    \centering
    \caption{Comparison between PSALM and RS in MG selection}
    \label{tab:rq1phase2}
    \resizebox{\columnwidth}{!}{
        \input{tables/tab-rq1phase2}
    }
\end{table}

\newcontent{
Figure~\ref{fig:rq1phase1} (see the rightmost vertical bar) shows that, for source test case selection, across all eight subject programs, the percentages of mutants that cause PSALM and RS to perform better are 47\% and 6\%, respectively.
Whereas for MG selection, Figure~\ref{fig:rq1phase2} (see the rightmost vertical bar) shows that, across all the eight subject programs, the percentages of mutants that cause PSALM and RS to perform better are 40\% and 11\%, respectively.
Here, we can see that, at the mutant level, PSALM largely outperforms RS in both source test case selection and MG selection.

Although the results shown in Tables~\ref{tab:rq1phase1} and~\ref{tab:rq1phase2}, as well as Figure~\ref{fig:rq1}, are largely consistent with our theoretical analysis in Section~\ref{sec:pssVSrs}, which states that PSALM is never inferior to RS in fault detection, a few exceptions are observed. 
For instance, Table~\ref{tab:rq1phase1} reports small negative improvements for ParseI and Conv, and Figure~\ref{fig:rq1} shows that, at the mutant level, RS performs better than PSALM for 6\% of mutants in source test case selection and 11\% of mutants in MG selection. 
These discrepancies appear inconsistent with the theoretical guarantee but can be explained by the following reason.

PSALM employs the BMA-MT algorithm to approximate the ideal proportional allocation of test cases or MGs to subdomains (see Section~\ref{sec:applyPSS4stcs}). 
In the theoretical analysis, these allocations are treated as real-valued, whereas in practice they must be integers. 
As a result, BMA-MT produces an integer approximation to the ideal ratios. 
When the subdomain sizes differ substantially, this approximation may deviate from the ideal proportional allocation and, consequently, lead to suboptimal empirical performance of PSALM.
In other words, these cases reflect a gap between the theoretical model and the practical approximation rather than a violation of the theoretical result.
}

\AnswertoRQ{\RQIndex{1}}{
Overall, PSALM achieves fault detection effectiveness that is equal to or higher than RS for source test case selection and MG selection.
}

\begin{figure*}[htbp]
  \centering
  % ===== (a) Selecting source test cases =====
  \begin{minipage}{0.7\linewidth}
    \centering
    \includegraphics[width=\linewidth]{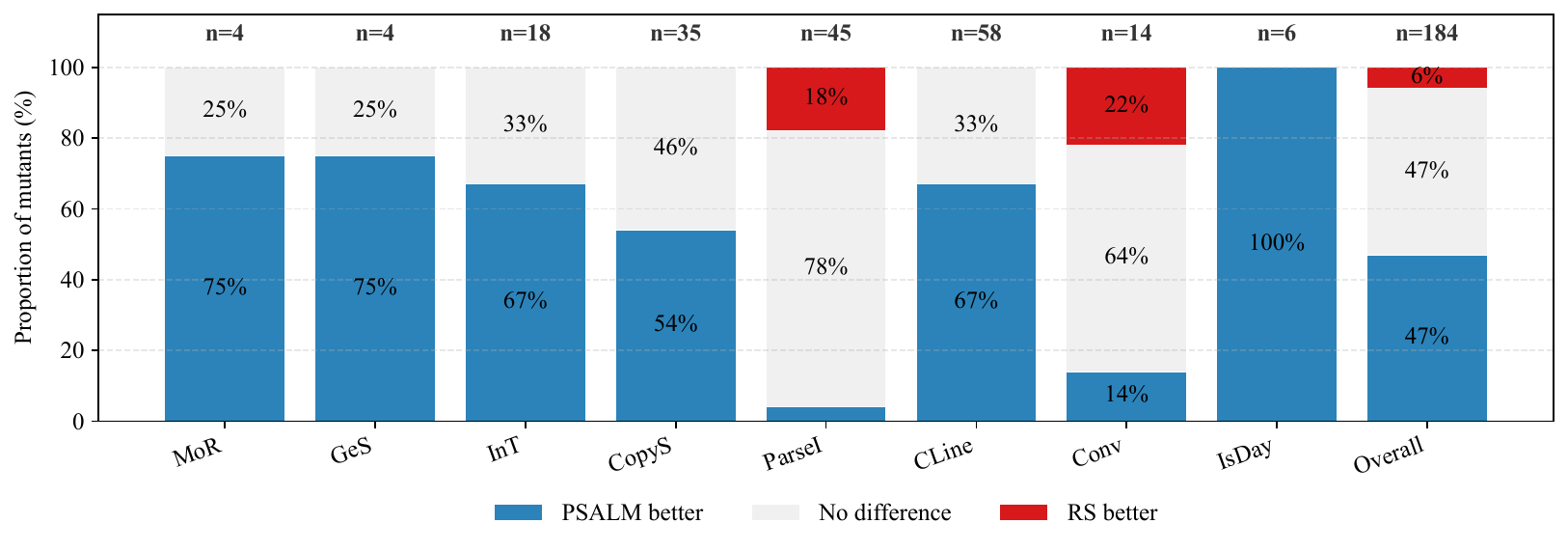}
    \subcaption{Source test case selection.}
    \label{fig:rq1phase1}
  \end{minipage}

  \vspace{6pt} % 控制上下间距

  % ===== (b) Selecting MGs =====
  \begin{minipage}{0.7\linewidth}
    \centering
    \includegraphics[width=\linewidth]{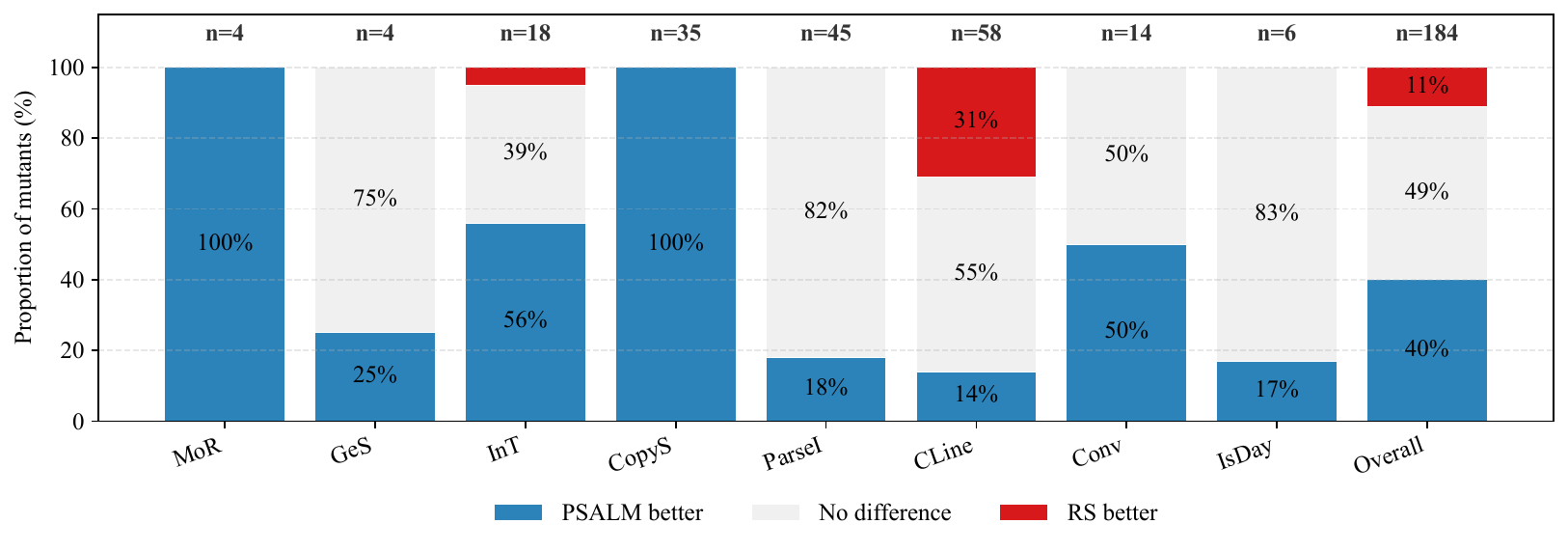}
    \subcaption{MG selection.}
    \label{fig:rq1phase2}
  \end{minipage}

  % ===== 主 caption =====
  \caption{Effectiveness comparison between PSALM and RS for source test case selection and MG selection.
  Each vertical bar shows the percentages of mutants leading to better performance, equal performance, and worse performance of PSALM over RS\@.
  The number above each vertical bar shows the total number of mutants generated for the respective subject program.}
  \label{fig:rq1}
\end{figure*}

\subsection{RQ2: Performance comparison of PSALM between source test case selection and MG selection}

% \newcontent{Having established that PSALM empirically outperforms RS in most cases, we next investigate whether the effectiveness of PSALM is higher for source test case selection or MG selection.
% According to our theoretical analysis, the two selections are guaranteed to achieve identical effectiveness only when the partitioning schemes are equivalent and each source test case generates the same number of follow-up test cases.
% When these conditions are not satisfied, the relative effectiveness of the two selections cannot be theoretically determined.
% Therefore, we conducted empirical studies to explore, under real settings, which selection yields higher fault-detection effectiveness.}
\newcontent{
Having established that PSALM empirically outperforms RS in most cases in the preceding subsection, we next examine how its effectiveness differs between source test case selection and MG selection in practical settings. 
Our theoretical analysis identifies the specific conditions under which applying PSALM to both selections would achieve identical fault detection effectiveness. 
However, these conditions may not always be met in practical settings. When this happens, the relative effectiveness of the two selections cannot be determined from theory. 
We therefore conducted an empirical study to examine, under real settings, which selection yields higher fault detection effectiveness.
}

\begin{table}[!ht]
    \ncolor
    % \footnotesize
    \centering
    \caption{Comparison between applying PSALM to source test case selection and to MG selection}
    \label{tab:rq2}
    \resizebox{\columnwidth}{!}{
        \input{tables/tab-rq2}
    }
\end{table}

\newcontent{
Table~\ref{tab:rq2} summarizes the comparison results between applying PSALM to source test case selection and MG selection.
% \footnote{\tt Please rearrange the columns of this tables so that the column corresponding to $P_{st}$ comes before the column corresponding to $P_{MG}$. Together with this arrangement you may need to update some figures as well (PL). \zzh{I positioned the MG column before st because both the p-value and A12 are calculated as MG over st. Reordering the columns could change the resulting A12 values. If you want to do this i can do it.}}
According to A12, for five (MoR, InT, CopyS, ParseI, and Conv) of the eight subject programs, applying PSALM to MG selection achieves a better performance in fault detection than for source test case selection (A12 $= 1.000$; $p < 0.05$).
% \footnote{\tt As far as I can remember, one reviewer explicitly asked us to compute A12, because he trusts that metric more. So, my discussion across all the 3 RQs will be based on this metric rather than the improvement percentage (PL).} 
For the remaining three subject programs (GeS, CLine, and IsDay), applying PSALM to source test case selection yields a better performance in fault detection (A12 $= 0.000$; $p < 0.05$).
Note that, in the above cases, the A12 statistics are either 1 or 0, indicating that our experiment results across all the 30 trials are largely consistent.
Regarding the percentage improvements reported in Table~\ref{tab:rq2}, they vary in magnitude across the subject programs, 
ranging from $-77.45\%$ to $72.41\%$.
% they do not influence the above dominance relationships as indicated by A12\@.
%%% Afterall, when compared with percentage improvement, A12 serves as a more reliable indication of which selection is more effective.
}

\newcontent{
% Figure~\ref{fig:rq2} summarizes the mutant-level comparison between applying PSALM to select MGs and to select source test cases.
We further verified whether or not the above observation also applies at the mutant level.
Among all the 184 mutants, 41.9\% of them cause applying PSALM to MG selection to achieve higher fault detection, whereas 24.5\% of them cause applying PSALM to source test case selection to perform better. 
The remaining 33.6\% show no statistically significant difference between applying PSALM to MG selection and source test case selection. 
These results are generally consistent with those observed in Table~\ref{tab:rq2}, i.e., applying PSALM to MG selection yields higher fault detection effectiveness than applying PSALM to source test case selection.
}

% We also performed a simulation study related to Proposition~\ref{prop:EffectivenessEquality}, by ensuring that each source test case is associated with only one MG and the partition scheme for the selection domain and the MG domain was equivalent (thereby satisfying the prerequisite conditions of Proposition~\ref{prop:EffectivenessEquality}).
% The results are shown in Table~\ref{tab:proposition3}.\zzh{need to revised later}
% For all the 26 mutants, not only their $P_p^{st}$s (mean $=$ 0.6876) and $P_p^\mathrm{MG}$s (mean $=$ 0.6891) were very close to each other, but their $p$-values were larger than 0.05, indicating that there were no significant differences between $P_p^{st}$s and $P_p^\mathrm{MG}$s. 
% These results are in line with Proposition~\ref{prop 3}, which states that when the prerequisite conditions are met, we have $P_p^{st} = P_p^\mathrm{MG}$.

\AnswertoRQ{\RQIndex{2}}{
In real settings, PSALM generally achieves higher fault detection effectiveness when applied to MG selection than to source test case selection.
}

\subsection{RQ3: Performance comparison among PSALM, ART, and MT-ART}
\newcontent{Besides RS, we further compare PSALM with two other baseline selection strategies commonly used in MT: ART and MT-ART (the rationales for selecting ART and MT-ART for benchmarking have been given in Section~\ref{sec:baselines}).
Tables~\ref{tab:rq3phase1} and~\ref{tab:rq3phase2} present the comparative results with respect to source test case selection and MG selection, respectively.
%Each table reports the mean P-measure values of both methods, the relative improvement achieved by PSALM, the corresponding $p$-values from statistical testing, and the Vargha–Delaney effect size.
% Specifically, in our experiment, ART was used for selecting source test cases, and MT-ART was used for selecting MGs.
% Basically, ART and MT-ART aim to improve fault detection by evenly distributing test cases and MGs, respectively, across the testing domain, whereas PSALM selects source test cases/MGs in proportion to the size of $D^{st}$/$D^\mathrm{MG}$.
}

\begin{table}[!ht]
    \ncolor
    % \footnotesize
    \centering
    \caption{Comparison between PSALM and ART in source test case selection}
    \label{tab:rq3phase1}
    \resizebox{\columnwidth}{!}{
        \input{tables/tab-rq3phase1}
    }
\end{table}

\begin{table}[!ht]
    \ncolor
    % \footnotesize
    \centering
    \caption{Comparison between PSALM and MT-ART in MG selection}
    \label{tab:rq3phase2}
    \resizebox{\columnwidth}{!}{
        \input{tables/tab-rq3phase2}
    }
\end{table}
\newcontent{
Table~\ref{tab:rq3phase1} shows that, when selecting source test cases, PSALM attains equal or higher $P$-measures than ART across all the eight subject programs.
Also, in six subject programs (GeS, InT, CopyS, CLine, Conv, and IsDay): (a)~the differences in $P$-measures are statistically significant ($p<0.05$), 
(b)~the effect sizes are large (A12 $\geqslant 0.712$), and % QQQ
% indicating that PSALM consistently enhances fault-detection effectiveness.
(c)~the percentage improvements are particularly prominent for InT and IsDay, reaching 221.18\% and 289.39\%, respectively.
%%% For \inlineCode{MoR} and \inlineCode{ParseI}, the differences are not statistically significant ($p>0.05$), suggesting that PSALM and ART achieve comparable fault-detection effectiveness in these programs.

Next, we turn to Table~\ref{tab:rq3phase2}\,%
\footnote{\,In this table, the $p$-values for GeS and ParseI are undefined, because the $P$-measures of all the 30 trials are identical.}.
This table shows that, with respect to MG selection, PSALM attains equal or higher $P$-measures than MT-ART in
seven (MoR, GeS, InT, CopyS, ParseI, Conv, and IsDay) out of the eight subject programs.
Furthermore, among these seven programs, five of them have very large A12 statistics ($\geqslant 0.998$). With respect to percentage improvement, InT exhibits the largest gain (402.44\%).

For GeS and ParseI, both have identical $P$-measures when using PSALM or MT-ART for MG selection (see Table~\ref{tab:rq3phase2}). 
A close examination found that, for every generated mutant associated with GeS and ParseI, it was almost always detectable or almost non-detectable, causing the $P$-measures of these mutants to be extremely close to 1 or 0. 
Since the values reported in Tables~\ref{tab:rq3phase2} are averages over all mutants, such highly polarized $P$-measures (i.e., very close to 1 or 0) also cause the overall averages to remain largely unchanged across different strategies.
In such cases, the $P$-measure becomes insensitive to the selection strategy, leading PSALM and MT-ART to exhibit nearly identical $P$-measure values. This also explains why almost all mutants in GeS and ParseI fall into the “No Difference" category in Figure~\ref{fig:rq3phase2}.

}

\begin{figure*}[!htbp]
  \centering
  % ===== (a) Selecting source test cases =====
  \begin{minipage}{0.7\linewidth}
    \centering
    \includegraphics[width=\linewidth]{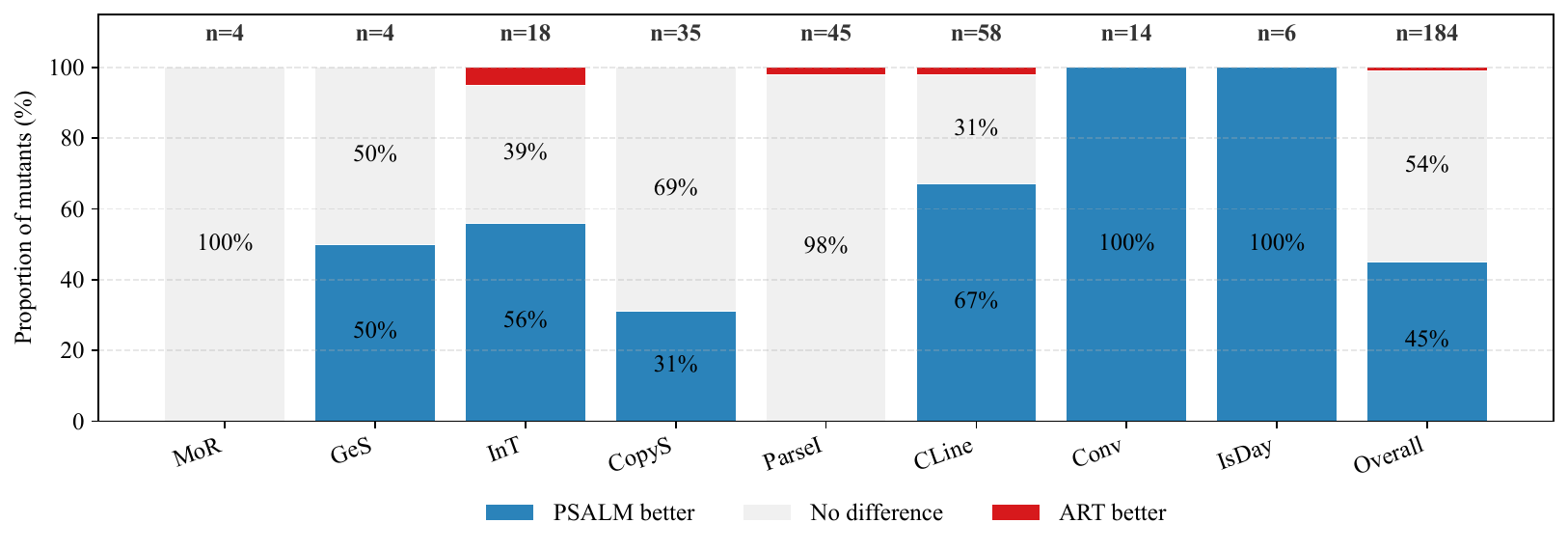}
    \subcaption{Selecting source test cases.}
    \label{fig:rq3phase1}
  \end{minipage}

  \vspace{6pt} % 控制上下间距

  % ===== (b) Selecting MGs =====
  \begin{minipage}{0.7\linewidth}
    \centering
    \includegraphics[width=\linewidth]{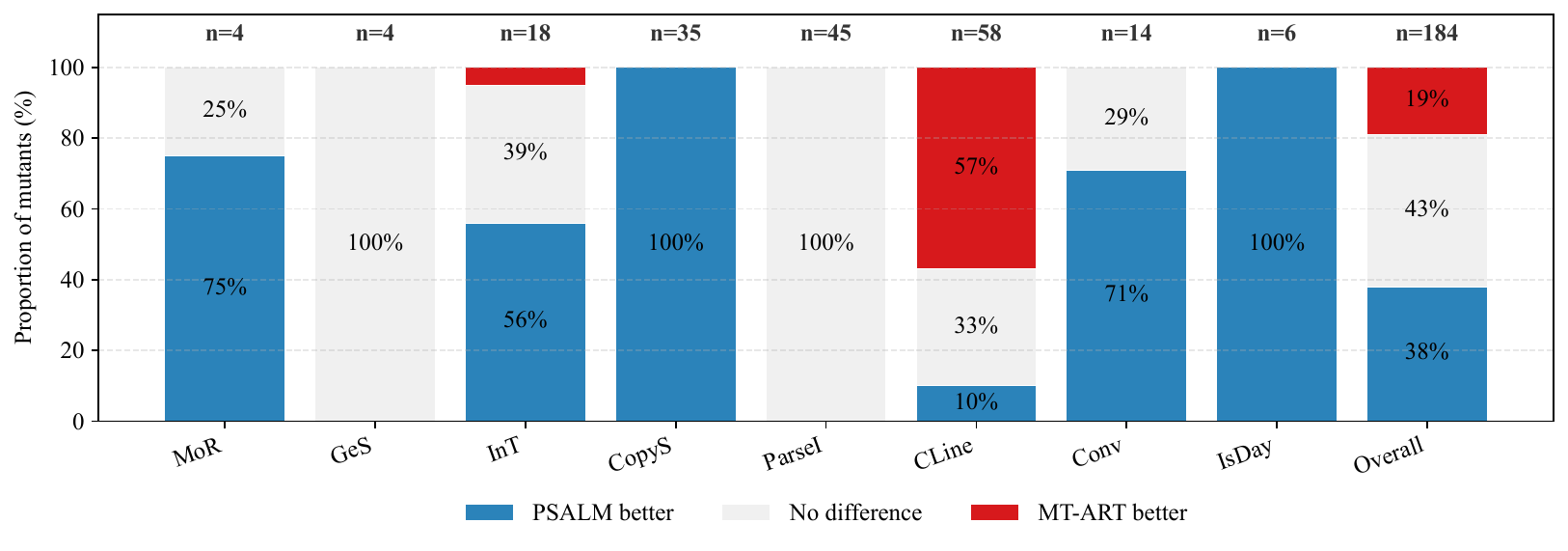}
    \subcaption{Selecting MGs.}
    \label{fig:rq3phase2}
  \end{minipage}

  % ===== 主 caption =====
  \caption{Effectiveness comparison among PSALM, ART (for source test case selection), and MT-ART (for MG selection).}
  \label{fig:rq3}
\end{figure*}

\newcontent{
Figure~\ref{fig:rq3} depicts the relative performance among PSALM, ART, and MT-ART at the mutant-level.
Each subfigure shows the proportion of mutants for which PSALM performs better, worse, or the same as ART/MT-ART across the subject programs.
For source test case selection (see Figure~\ref{fig:rq3phase1}), PSALM performs better than ART on 45\% of mutants overall, whereas ART performs better only on 1\% of mutants.
For the remaining 54\% of mutants, PSALM and ART show identical fault detection effectiveness. 
For MG selection (see Figure~\ref{fig:rq3phase2}), PSALM outperforms MT-ART on 38\% of mutants overall, while MT-ART performs better only on 19\% of mutants.
The remaining 43\% of mutants show no difference between PSALM and MT-ART\@.

These mutant-level distributions are consistent with the program-level observations in Tables~\ref{tab:rq3phase1} and~\ref{tab:rq3phase2}. 
In particular, programs such as GeS and ParseI exhibit a large proportion of ``No Difference" mutants, which aligns with our earlier analysis that their $P$-measures remain unchanged across strategies.
}

\AnswertoRQ{\RQIndex{3}}{In terms of fault detection effectiveness, PSALM generally outperforms ART and MT-ART in source test case selection and MG selection, respectively.}

%% file: tables/tab-rq1phase1.tex
\begin{tabular}{lccccc}
\toprule
Program & $\overline{P}_{\mathrm{PSALM}}$ & $\overline{P}_{\mathrm{RS}}$ & Improvement (\%) & $p$-value & A12 \\
\midrule
MoR & 0.9124 & 0.8996 & 1.42 & $<0.05$ & 0.987 \\
GeS & 0.5038 & 0.5014 & 0.49 & $<0.05$ & 0.970 \\
InT & 0.2577 & 0.0766 & 236.27 & $<0.05$ & 1.000 \\
\midrule
CopyS & 0.8851 & 0.8793 & 0.66 & $<0.05$ & 0.999 \\
ParseI & 0.9994 & 0.9995 & $-$0.01 & $<0.05$ & 0.251 \\
CLine & 0.7317 & 0.4277 & 71.06 & $<0.05$ & 1.000 \\
Conv & 0.8754 & 0.8966 & $-$2.36 & $<0.05$ & 0.000 \\
IsDay  & 0.2793 & 0.0633 & 341.13 & $<0.05$ & 1.000 \\
\bottomrule
\end{tabular}

%% file: tables/tab-rq1phase2.tex
\begin{tabular}{lccccc}
\toprule
Program & $\overline{P}_{\mathrm{PSALM}}$ & $\overline{P}_{\mathrm{RS}}$ & Improvement (\%) & $p$-value & A12 \\
\midrule
MoR & 0.9952 & 0.7849 & 26.79 & $<0.05$ & 1.000 \\
GeS & 0.5000 & 0.4995 & 0.10 & $<0.05$ & 0.967 \\
InT & 0.4444 & 0.0829 & 436.14 & $<0.05$ & 1.000 \\
\midrule
CopyS & 0.9769 & 0.9431 & 3.59 & $<0.05$ & 1.000 \\
ParseI & 1.0000 & 0.9999 & 0.01 & $<0.05$ & 0.998 \\
CLine & 0.4410 & 0.4372 & 0.87 & $<0.05$ & 1.000 \\
Conv & 0.9022 & 0.8852 & 1.92 & $<0.05$ & 1.000 \\
IsDay & 0.0630 & 0.0622 & 1.28 & $0.309$ & 0.598 \\

\bottomrule
\end{tabular}

%% file: tables/tab-rq2.tex
\begin{tabular}{lccccc}
\toprule
Program & $\overline{P}_{\mathrm{MG}}$ & $\overline{P}_{\mathrm{st}}$ & Improvement (\%) & $p$ & A12 \\
\midrule
MoR & 0.9952 & 0.9124 & 9.07 & $<0.05$ & 1.000 \\
GeS & 0.5000 & 0.5038 & $-$0.75 & $<0.05$ & 0.000 \\
InT & 0.4444 & 0.2577 & 72.41 & $<0.05$ & 1.000 \\
\midrule
CopyS & 0.9769 & 0.8851 & 10.38 & $<0.05$ & 1.000 \\
ParseI & 1.0000 & 0.9994 & 0.06 & $<0.05$ & 1.000 \\
CLine & 0.4410 & 0.7317 & $-$39.73 & $<0.05$ & 0.000 \\
Conv & 0.9022 & 0.8754 & 3.06 & $<0.05$ & 1.000 \\
IsDay & 0.0630 & 0.2793 & $-$77.45 & $<0.05$ & 0.000 \\

\bottomrule
\end{tabular}

%% file: tables/tab-rq3phase1.tex
\begin{tabular}{lccccc}
\toprule
Program & $\overline{P}_{\mathrm{PSALM}}$ & $\overline{P}_{\mathrm{ART}}$ & Improvement (\%) & $p$-value & A12 \\
\midrule
MoR & 0.9124 & 0.9113 & 0.12 & $0.537$ & 0.572 \\
GeS & 0.5038 & 0.5028 & 0.20 & $<0.05$ & 0.769 \\
InT & 0.2577 & 0.0802 & 221.18 & $<0.05$ & 1.000 \\
\midrule
CopyS & 0.8851 & 0.8838 & 0.14 & $<0.05$ & 0.712 \\
ParseI & 0.9994 & 0.9994 & 0.00 & $0.646$ & 0.474 \\
CLine & 0.7317 & 0.4249 & 72.18 & $<0.05$ & 1.000 \\
Conv & 0.8754 & 0.4381 & 99.84 & $<0.05$ & 1.000 \\
IsDay & 0.2793 & 0.0717 & 289.39 & $<0.05$ & 1.000 \\
\bottomrule
\end{tabular}

%% file: tables/tab-rq3phase2.tex
\begin{tabular}{lccccc}
\toprule
Program & $\overline{P}_{\mathrm{PSALM}}$ & $\overline{P}_{\mathrm{MT-ART}}$ & Improvement (\%) & $p$-value & A12 \\
\midrule
MoR & 0.9952 & 0.9145 & 8.82 & $<0.05$ & 1.000 \\
GeS & 0.5000 & 0.5000 & 0.00 & $-$ & 0.500 \\
InT & 0.4444 & 0.0884 & 402.44 & $<0.05$ & 1.000 \\
\midrule
CopyS & 0.9769 & 0.8989 & 8.67 & $<0.05$ & 1.000 \\
ParseI & 1.0000 & 1.0000 & 0.00 & $-$ & 0.500 \\
CLine & 0.4410 & 0.4525 & $-$2.54 & $<0.05$ & 0.000 \\
Conv & 0.9022 & 0.8632 & 4.53 & $<0.05$ & 1.000 \\
IsDay & 0.0630 & 0.0496 & 27.05 & $<0.05$ & 0.998 \\
\bottomrule
\end{tabular}

%% file: section/sec-threats2validity.tex
\section{Threats to Validity}
\label{sec:threatsValidity}

\textit{Subject program selection:} 
Our study involved \newcontent{eight subject programs in total. 
These include three classical benchmark programs from~\cite{ChanChen1996} and five additional real-world methods extracted from Defects4J projects.} 
The three benchmark programs cover distinct partition scenarios: MoR features subdomain sizes that do not conform to simple integer ratios; InT has subdomains whose sizes vary substantially; and GeS contains an extremely large number of subdomains. 
\newcontent{The five additional real-world programs provide evidence of the effectiveness of our method on realistic software.}
Hence, we argue that the \newcontent{eight} subject programs have provided useful insights into the use of {PSALM} for selecting source test cases and MGs.

\textit{Mutant generation and selection:}
Due to the absence of available mutants for the subject programs, we used the \textsc{Mutmut} tool \newcontent{together with the \textsc{Major} framework} to generate mutants for our study.
To ensure the feasibility and reliability of the experiment, we examined all generated mutants and removed those that were ``non-executable'' \newcontent{as well as mutants that were clearly equivalent, trivial, or subsumed, in order to avoid redundancy and retain only representative fault-revealing mutants}.
These removals may reduce the overall breadth of the mutant set.
Nevertheless, the observed results were highly consistent across different subject programs \newcontent{based on the 184 selected mutants}, providing strong support for the validity of our findings.

\textit{MR identification:}
Due to the absence of publicly available MRs for the three subject programs, the MRs used in our study \newcontent{were manually defined by us based on their documented functional properties}. 
Obviously, the more MRs we used, the more comprehensive the study would be. 
Nevertheless, using a limited MR set does not invalidate our comparison between \newcontent{PSALM and the baseline selection strategies}, because the empirical observations were consistent with the underlying mathematical guarantees. 
\newcontent{All manually defined MRs were discussed by multiple authors to ensure their correctness and relevance.}
\newcontent{In addition, all implementations used in our study have been made publicly available to facilitate independent inspection and replication.}

% \textit{Randomness and statistical analysis:}
% \newcontent{
% Our study involves selection strategies (RS, ART, MT-ART, and PSALM) that contain inherent randomness.
% Although all experiments were repeated sufficiently many times and the results were analyzed using statistical tests and effect sizes, random fluctuations cannot be completely eliminated.
% Different random seeds or additional repetitions may lead to slight numerical variations, but they do not affect the overall conclusions observed in our study.
% }

%% file: section/sec-relatedWork.tex
\section{Related Work}
\label{sec:relatedWork}

\noindent\textbf{Partition testing and PSS\@.}
\newcontent{Partition testing has long been studied as a systematic alternative to random testing.}
Many existing studies on partition testing involve conducting empirical studies to investigate the fault detection effectiveness of partition testing over random testing~\cite{Vagoun1996, Sinaga2018, ChenYu1996}.  
\newcontent{To provide a theoretical foundation, Chen and Yu analyzed partition testing using a formal probabilistic model and introduced PSS~\cite{ChenYu1994, ChenYu1996, ChenYu1996a}.}
\newcontent{Subsequent work further established PSS as a necessary and sufficient condition for achieving universal safeness and refined its practical guidelines~\cite{ChanChen1996, ChenTse2001, LeungChen1999, LeungChen2000, ChenYu2000, ChenYu2001}.}
\newcontent{Beyond traditional testing, related studies have examined partition testing in broader contexts and proposed dynamic or adaptive variants, including dynamic partitioning, random testing with dynamically updated test profiles, and ART~\cite{SunFu2021, Vagoun1996, CaiJing2005, CaiHu2009, SunDai2019}.}
\newcontent{However, none of these studies analyzed partition testing in the MT context, where the selection domains, test units, and fault distributions differ fundamentally from traditional testing.}
\newcontent{As a result, it remains unknown whether the effectiveness of PSS over RS, as established in traditional testing, still holds in the context of MT\@.}

%\vspace{0.5ex} % adhoc
%\noindent\textbf{MT\@.}
%It has been extensively studied as an effective approach to alleviating the oracle problem~\cite{ChenCheung1998, ChanChen1998, Chen2015, SeguraFraser2016}.  
%Existing studies have mainly concentrated on two directions.  
%First, prior work has extensively explored the identification, construction, and composition of MRs~\cite{SeguraFraser2016, SunFu2021, QiuZheng2022, AyerdiTerragni2024, NolascoMolina2024, ShinPastore2024}.  
%Second, many studies have examined the applicability and empirical effectiveness of MT across different domains, including search engines, embedded systems, web services, machine learning, cyber–physical systems, and quantum computing~\cite{XiaoLiu2022, XieZhang2020, ZhouXiang2016, SeguraFraser2016, MandrioliShin2025}.  
%By contrast, systematic methods for source test case selection and MG selection have received far less attention, even though they directly affect the coverage of the MT process and thus its overall fault detection effectiveness~\cite{BarusChen2016, SahaKanewala2018, HuiWang2021}.

% \vspace{-0.5ex} % adhoc
\noindent\textbf{Source test case selection and MG selection in MT\@.}
\newcontent{Several studies have explored how test case selection may influence the fault detection effectiveness of MT, although these efforts remain limited in scope.
Barus et al.~\cite{BarusChen2016} applied ART to select source test cases and observed improvements in fault detection effectiveness. 
Saha and Kanewala~\cite{SahaKanewala2018} examined coverage-guided source test case generation strategies and reported similar benefits. 
Sun et al.~\cite{SunLiu2022} further proposed path-directed generation and prioritization of source test cases to improve the efficiency of MT\@. }
\newcontent{In addition to source test case selection, several approaches have been proposed for MG selection.
MD-ART, and its extension MT-ART, apply ART to MT by selecting MGs using distance measures over the source and follow-up test cases~\cite{HuiWang2021}.}
\newcontent{More recently, SFIDMT-ART introduced distance criteria that explicitly account for the domain difference between the source and follow-up input domains to select MGs~\cite{YingTowey2024}.}
\newcontent{However, these studies focus primarily on diversity or structural coverage and do not analyze the fault-detection effectiveness of MT selection strategies themselves.}
\newcontent{In particular, none of them connects MT test selection with the theoretical results of partition testing or examines whether the effectiveness properties guaranteed by PSS still hold in MT\@.}

%% file: section/sec-conclusion.tex
\section{Conclusion and Future Work}
\label{sec:conclusion}

In this paper, we have introduced our PSALM method, involving the use of proportional sampling strategy (PSS) for selecting source test cases and metamorphic groups (MGs), with a view to improving the fault detection effectiveness of metamorphic testing (MT). Associated with the PSALM method, we have performed theoretical analysis to:
(a)~compare randomly selecting source test cases/MGs with \newcontent{using our PSALM method for such selection}, and 
(b)~compare the use of \newcontent{our PSALM method} between source test case selection and MG selection.
More specifically, we have formally proved that:
(i)~regardless of the partition scheme, we have $P_p^{st} \geqslant P_r^{st}$;
(ii)~regardless of the partition scheme, we have $P_p^\mathrm{MG} \geqslant P_r^\mathrm{MG}$;
(iii)~when the prerequisite conditions are met, we have $P_p^{st} = P_p^\mathrm{MG}$; and
(iv)~when the prerequisite conditions are \newcontent{not} met, $P_p^{st}$ can be greater than, equal to, or less than $P_p^\mathrm{MG}$.

\newcontent{
To evaluate the effectiveness of our PSALM method in realistic MT settings, we have conducted an experiment involving eight subject programs, 184 mutants, and a set of manually defined metamorphic relations (MRs). 
Overall, the experimental results are consistent with our theoretical analysis. 
The results further indicate that PSALM generally achieves higher fault detection effectiveness than ART and MT-ART, and that applying PSALM to source test case selection incurs lower computational overhead than applying it to MG selection in practice.
}

In view of our encouraging result, we plan to apply our PSALM method to different application domains (e.g., embedded systems, compilers, machine learning, web services, online search engines, autonomous car driving, and large language models) to further investigate the effectiveness of PSALM\@. 
\newcontent{
Another direction for future work is to investigate how source test case selection and MG selection can be jointly applied in a unified selection process. Different ways of integrating these two levels of selection may lead to different levels of effectiveness.
}

% \newcontent{
% We have discussed PSALM, a strategy to adapt PSS to source test case selection and MG selection in MT, with a view to improving its fault detection effectiveness. We generalized the theoretical foundations of PSS from traditional testing to MT, where the test unit is an MG rather than an individual test case, and proved that PSALM is never inferior to RS for source test case selection and MG selection (when certain conditions are met), regardless of the partition scheme used. We also identified the conditions under which the fault detection effectiveness of PSALM for source test case selection and MG selection becomes identical.

% We conducted a series of experiments involving eight subject programs. The results show that PSALM outperforms RS in most settings and generally achieves higher fault detection effectiveness when compared with ART and MT-ART\@. Our experiments also show that applying PSALM to MG selection is generally more effective than applying it to source test case selection in practice.

% Future work may focus on applying PSALM to more application domains to further evaluate its generality. In addition, the combination of source test case selection and MG selection deserves further investigation, as different ways of using these two selections may lead to different levels of effectiveness. Finally, exploring other sampling methods may provide additional improvements to MT\@.
% }

%% file: section/sec-acknowledgment.tex
\section*{Acknowledgment}

This work was supported in part by the National Natural Science Foundation of China under Grant No.~62372021\@.
The authors would like to thank Professor T.\,Y.~Chen at Swinburne University of Technology, Australia, for his valuable review and support in this research project.

%% file: component/com-ref.tex
\bibliographystyle{IEEEtran}
\bibliography{ref}